\documentclass[10pt,twocolumn,twoside,romanappendices]{IEEEtran}

\newcommand{\diag}{\mathop{\rm diag}}
\newcommand{\Tr}{\mathop{\rm tr}}
\newcommand{\Var}{\mathop{\rm Var}}
\makeatletter
\newcommand{\vast}{\bBigg@{4}}
\newcommand{\vastt}{\bBigg@{6}}
\makeatother
\usepackage{color}
\usepackage{multirow}
\usepackage{graphicx}
\usepackage{epstopdf}
\usepackage[cmex10]{amsmath}
\usepackage{amssymb}

\usepackage{multirow}
\usepackage{algpseudocode}
\usepackage{amsthm}
\usepackage{cite}
\usepackage{bbm}
\theoremstyle{remark}
\newtheorem{rem}{Remark}
\newtheorem{lem}{Lemma}
\newtheorem*{ddef}{Definition}

\newtheorem*{cor}{Corollary}
\newtheorem{prop}{Proposition}

\makeatletter
\g@addto@macro\th@remark{\thm@headpunct{\normalfont:}}
\makeatother
\makeatletter
\newcommand{\distas}[1]{\mathbin{\overset{#1}{\kern\z@\sim}}}%
\newsavebox{\mybox}\newsavebox{\mysim}
\newcommand{\distras}[1]{%
  \savebox{\mybox}{\hbox{\kern3pt$\scriptstyle#1$\kern3pt}}%
  \savebox{\mysim}{\hbox{$\sim$}}%
  \mathbin{\overset{#1}{\kern\z@\resizebox{\wd\mybox}{\ht\mysim}{$\sim$}}}%
	}

\allowdisplaybreaks

\begin{document}

\title{Distributed Spatial Multiplexing Systems with Hardware Impairments and Imperfect Channel Estimation under Rank-$1$ Rician Fading Channels}

\author{Nikolaos~I.~Miridakis, Theodoros~A.~Tsiftsis,~\IEEEmembership{Senior Member,~IEEE}\\and~Corbett~Rowell,~\IEEEmembership{Senior Member,~IEEE}
\thanks{N. I. Miridakis is with the Department of Computer Engineering, Piraeus University of Applied Sciences, 12244 Aegaleo, Greece (e-mail: nikozm@unipi.gr).}
\thanks{T. A. Tsiftsis is with the School of Engineering, Nazarbayev University, 010000 Astana, Kazakhstan and with the Department of Electrical Engineering, Technological Educational Institute of Central Greece, 35100 Lamia, Greece (e-mails: theodoros.tsiftsis@nu.edu.kz, tsiftsis@teiste.gr).}
\thanks{C. Rowell is with the Dept. of Electronic \& Computer Engineering, Hong Kong University of Science \& Technology, Hong Kong and with Rohde \& Schwarz, Munich, Germany (email: corbett.rowell@gmail.com).}}

\markboth{}%
{}

\maketitle

\begin{abstract}
The performance of a multiuser communication system with single-antenna transmitting terminals and a multi-antenna base-station receiver is analytically investigated. The system operates under independent and non-identically distributed rank-$1$ Rician fading channels with imperfect channel estimation and residual hardware impairments (compensation algorithms are assumed, which mitigate the main impairments) at the transceiver. The spatial multiplexing mode of operation is considered where all the users are simultaneously transmitting their streams to the receiver. Zero-forcing is applied along with successive interference cancellation (SIC) as a means for efficient detection of the received streams. New analytical closed-form expressions are derived for some important performance metrics, namely, the outage probability and ergodic capacity of the entire system. Both the analytical expressions and simulation results show the impact of imperfect channel estimation and hardware impairments to the overall system performance in the usage scenarios of massive MIMO and mmWave communication systems.
\end{abstract}


\IEEEpeerreviewmaketitle

\section{Introduction}
\IEEEPARstart{S}{patial} multiplexing represents one of the most prominent techniques used for multiple input-multiple output (MIMO) transmission systems \cite{j:gesbert}. In order to reduce the computational complexity at the receiver, linear detectors are usually employed such as  zero-forcing (ZF), or a simplified nonlinear yet capacity-efficient method of successive interference cancellation (SIC). These two techniques can be combined (ZF-SIC) in order to provide an appropriate tradeoff between performance and computational complexity (e.g., see \cite{j:MiridakisSurvey} and references therein). The popularity of ZF-SIC is mainly due to the fact that it achieves a high spectral efficiency and a substantial capacity gain, e.g., the V-BLAST approach \cite{golden1999detection}. This is achieved by the sequential detection/decoding of each stream, while it cancels the intra-stream interference at the same time; an efficient technique for MIMO transmission systems. 

Performance assessment of ZF and ZF-SIC has been extensively studied in the technical literature to date \cite{j:Toboso2014,j:MiridakisKaragiannidis2014,j:JiangVaranasi2008,j:JiangVaranasi2011,j:Loyka2006,j:LoykaGagnon2004,j:WCLMiridakisVergados2013,j:OzyurtTorlak2013}. All of these studies assumed perfect channel estimation at the receiver during communication and non-impaired hardware at the transceiver; this is an ideal and a rather overoptimistic scenario for practical communication systems.  Wireless transceiver hardware is usually subject to many impairments: I/Q imbalance, phase noise, and high-power amplifier nonlinearities \cite{schenk2008rf}. These impairments are typically mitigated with the aid of certain compensation algorithms. Nevertheless, inadequate compensation mainly due to the imperfect parameter estimation and/or time variation of the hardware characteristics may result to residual impairments, which are added to the transmitted/received signal \cite{Bjornsonhi2014}. It has been verified from both analytical and experimental results, e.g., \cite{schenk2008rf,zetterberg2011experimental}, that residual impairments can be modeled as additive noise-like signals with certain properties. In addition, an erroneous channel estimation may occur due to imperfect feedback/feedforward signaling and/or rapid channel variations. Several recent research works have investigated \emph{non-ideal} system configurations, governed by either impaired hardware at the transceiver \cite{c:StuderWenk2010,c:ZhangMatthaiouBjornson2014} or imperfect estimates of the channel gains \cite{j:WangMurch2007,j:Narasimhan2005} focusing on the limited case of Rayleigh channel fading. Nonetheless, the joint impact of hardware impairments and channel estimation errors for ZF(-SIC) systems has not been studied into the open technical literature so far.

Rician fading is more general than Rayleigh and more realistic for practical communication systems \cite{kyosti2007winner}. This occurs because the scenario of line-of-sight (LOS) signal propagation is included in the Rician fading model. Two cornerstone paradigms of the promising $5$G deployments, massive MIMO \cite{j:RusekPersson2013} and millimeter-wave (mmWave) communications \cite{j:RappaportShu2013}, rely mostly on LOS (or near-LOS) signal propagation (e.g., see \cite{j:AndrewsBuzzi2014,j:JinMatthaiou2015,j:PengWang2015} and references therein). Recent relevant works studying the performance of spatial multiplexing systems in Rician fading channels, e.g., \cite{j:Siriteanurician2014,j:siriteanuexact2014,j:SiriteanuMiyanaga2012,j:Zhang2014,j:XueSellathurai2015}, assumed either perfect channel estimation or non-impaired hardware at the transceiver. Pioneer works in \cite{j:Siriteanurician2014,j:siriteanuexact2014,j:SiriteanuMiyanaga2012} analyzed the performance of ZF for Rician fading channels including spatial correlation at the transmitter side for the general case of arbitrary (finite) ranges of the antenna array at the transceiver. Studies in \cite{j:Zhang2014,j:XueSellathurai2015} focused on massive MIMO (i.e., deriving asymptotic performance limits) and virtual MIMO systems via multiway relaying transmission, correspondingly.

{\color{black}Current work presents a unified analytical performance study of ZF and ZF-SIC receivers for non-ideal transmission systems, operating under rank-$1$ Rician fading channels with independent and non-identically distributed (i.n.i.d.) statistics for each transmitter.\footnote{Minimum mean-squared error (MMSE) represents another linear detection scheme where MMSE outperforms ZF at the cost of a higher computational burden since the noise variance is required in this case. Their performance gap becomes marginal in a moderately medium-to-high signal-to-noise ratio (SNR) regime \cite{j:JiangVaranasi2011}. For mathematical tractability, we focus on ZF/ZF-SIC herein and we leave the investigation of MMSE for a future work.} This is a suitable model for distributed-MIMO systems, where the involved users maintain arbitrary distances with the receiver and themselves (e.g., consider a heterogeneous cellular network). The rank-$1$ channel model limitation implies that at most one of the received streams experiences Rician fading, whereas all the remaining ones undergo Rayleigh channel fading conditions. Higher rank channel conditions correspond to more LOS-propagated signals; nonetheless, small antenna apertures and large transmitter-receiver distances are likely to yield rank-$1$ channels \cite{j:Siriteanurank0ne}. Further, rank-$1$ channel conditions can be justified as relevant in heterogeneous networking infrastructures (e.g., opportunistic scheduling of femto-cells within macro-cell deployments) \cite{j:SaquibHossain,j:siriteanuexact2014}, and/or in mmWave communications using beamspace MIMO transmission design \cite{j:BradyBehdad}. The ideal communication scenario with perfect hardware and channel estimates is only considered as a special case. A direct applicability of the presented framework can be found in any wireless communication system with LOS or near-LOS transmission using non-ideal equipment. Particular emphasis is given in MIMO systems with large number of antennas, where  the vast yet finite antenna array consists of low-cost (non-ideal) hardware.}

The contributions of this work are summarized as follows:
\begin{itemize}
	\item Novel closed-form expressions are derived with respect to the outage probability of each transmitted stream.
	\item A new analytical expression for the ergodic capacity of the Rician-faded stream in terms of a single infinite series representation and a corresponding new closed-form expression for the remaining Rayleigh-faded streams (and the entire system sum-capacity) are presented. The aforementioned analytical expression of the Rician-faded signal, although in a non-closed form, is straightforward, not computationally complex and is formed by rapidly converging series.
	\item These formulae tightly approximate the performance of the general case, while they become exact under perfect channel estimation.
	\item The derived results are more computationally efficient than existing methods, such as numerical manifold integrations and exhaustive simulations.
\end{itemize}

The rest of this paper is organized as follows: Section \ref{System Model} describes the considered system model. In Section \ref{Statistics of SNR} some key statistical properties of the received signal-to-noise ratio (SNR) for each stream are analyzed. The derivation of the considered system metrics of outage probability and ergodic capacity are provided in Section \ref{Performance Metrics}, while some relevant numerical results are illustrated in Section \ref{Numerical Results}. Finally, Section \ref{Conclusions} concludes the paper.

\emph{Notation}: Vectors and matrices are represented by lowercase bold typeface and uppercase bold typeface letters, respectively. $\mathbf{X}^{-1}$ is the inverse of $\mathbf{X}$ and $\mathbf{x}_{i}$ denotes the $i$th coefficient of $\mathbf{x}$. The coefficient in the $i$th row and $j$th column of $\mathbf{X}$ is presented as $[\mathbf{X}]_{ij}$. A diagonal matrix with entries $x_{1},\cdots,x_{n}$ is defined as $\diag\{x_{i}\}^{n}_{i=1}$. The superscripts $(\cdot)^{T}$ and $(\cdot)^{\mathcal{H}}$ denote transposition and Hermitian transposition, respectively, $\otimes$ is the Kronecker product between matrices, $\|\cdot\|$ corresponds to the vector Euclidean norm, while $|\cdot|$ represents the absolute (scalar) value. Trace and determinant of $\mathbf{X}$ are, respectively, given as $\Tr[\mathbf{X}]$ and $\det[\mathbf{X}]$. In addition, $\mathbf{I}_{v}$ stands for the $v\times v$ identity matrix, $\mathbb{E}[\cdot]$ is the expectation operator, $\Var[\cdot]$ represents statistical variance, $\overset{\text{d}}=$ represents equality in probability distributions, $\overset{\text{d}}\approx$ denotes almost equality in probability distributions and $\text{Pr}[\cdot]$ returns probability. Also, $f_{X}(\cdot)$ and $F_{X}(\cdot)$ represent probability density function (PDF) and cumulative distribution function (CDF) of the random variable (RV) $X$, respectively. Complex-valued Gaussian RVs with mean $\mu$ and variance $\sigma^{2}$, non-central chi-squared and (central) chi-squared RVs are denoted, respectively, as $\mathcal{CN}(\mu,\sigma^{2})$, $\mathcal{X}^{2}_{v}(u)$ and $\mathcal{X}^{2}_{v}$ with $v$ degrees-of-freedom (DoF) and $u$ is the non-centrality parameter. Also, $\text{Beta}(a,b)$ stands for the (central) Beta distribution with $a$ and $b$ as scale parameters. $\mathcal{CW}_{M}(N,\mathbf{R})$ is a complex-valued central Wishart distribution with dimension $M$, DoF $N$ and scale matrix $\mathbf{R}$. Moreover, $\Gamma(\cdot)$ denotes the Gamma function \cite[Eq. (8.310.1)]{tables}, $B(\cdot,\cdot)$ is the Beta function \cite[Eq. (8.384.1)]{tables}, $\Gamma(\cdot,\cdot)$ is the upper incomplete Gamma function \cite[Eq. (8.350.2)]{tables}, $\gamma(\cdot,\cdot)$ is the lower incomplete Gamma function \cite[Eq. (8.350.1)]{tables}, while $(\cdot)_{p}$ is the Pochhammer symbol with $p \in \mathbb{N}$ \cite[p. xliii]{tables}. $I_{n}(\cdot)$ represents the $n$th order modified Bessel function of the first kind \cite[Eq. (8.445)]{tables}, ${}_1F_{1}(\cdot,\cdot;\cdot)$ is the Kummer's confluent hypergeometric function \cite[Eq. (9.210.1)]{tables}, ${}_2F_{1}(\cdot,\cdot,\cdot;\cdot)$ is the Gauss hypergeometric function \cite[Eq. (9.100)]{tables}, $Q_{\nu}(\cdot,\cdot)$ is the generalized $\nu$th order Marcum-$Q$ function \cite{b:marcum}, and $Q_{\mu,\nu}(\cdot,\cdot)$ stands for the standard Nuttall-$Q$ function \cite[Eq. (86)]{nuttall1972some}. Finally, $\mathcal{O}(\cdot)$ is the Landau symbol, i.e., $f(x)=\mathcal{O}(g(x))$, when $|f(x)|\leq v |g(x)| \ \forall x\geq x_{0}$, $\left\{v,x_{0}\right\} \in \mathbb{R}$.

\section{System Model}
\label{System Model}
Consider a wireless communication system with $M$ single-antenna transmitters and a receiver equipped with $N\geq M$ antennas.\footnote{From the analysis presented hereafter, the classical single-user communication scenario with $M$ co-located transmit antennas is included as a special case.} The receiver performs channel estimation, while full-blind transmitters are assumed. The spatial multiplexing mode of operation is implemented, where $M$ independent data streams are simultaneously transmitted by the corresponding nodes. A suboptimal yet efficient linear detection scheme of ZF is adopted which is performed at the receiver using successive decoding or SIC. The actual and the detected signal at the receiver are, respectively, defined as
\begin{align}
\mathbf{y}=\mathbf{H} (\mathbf{s}+\mathbf{n}_{T})+\mathbf{n}_{R}+\mathbf{w},
\label{eq11}
\end{align}
and
\begin{align}
\hat{\mathbf{y}}=\hat{\mathbf{H}} (\mathbf{s}+\mathbf{n}_{T})+\mathbf{n}_{R}+\mathbf{w},
\label{eq1}
\end{align}
where $\mathbf{y} \in \mathbb{C}^{N\times 1}$, $\hat{\mathbf{y}} \in \mathbb{C}^{N\times 1}$, $\mathbf{H}\in \mathbb{C}^{N\times M}$, $\hat{\mathbf{H}}\in \mathbb{C}^{N\times M}$, $\mathbf{s} \in \mathbb{C}^{M\times 1}$ and $\mathbf{w} \in \mathbb{C}^{N\times 1}$ are the received signal, the detected signal, the actual flat-fading\footnote{A narrowband communication scenario is assumed via mutually independent subcarrier transmissions, as in current long-term evolution (LTE) systems. This strategy can also be used for wideband communication scenarios (e.g., mmWave), by establishing a vast range of multiple subcarriers, and hence facilitating the wideband communication with the aid of multiple virtual narrowband transmissions \cite{j:Heath2016,j:Barati2015,j:Alkhateeb2014,j:Ayach2014}.} channel matrix, the estimated channel matrix, the transmitted signal, and the additive white Gaussian noise (AWGN), respectively. Let $\mathbf{w}\overset{\text{d}}=\mathcal{CN}(\mathbf{0},N_{0}\mathbf{I}_{N})$ with $N_{0}$ representing the AWGN variance and $\mathbb{E}[\mathbf{s}\mathbf{s}^{\mathcal{H}}]=p \mathbf{I}_{M}$ with $p$ denoting the transmit power per antenna. The hardware impairments occurring at the transmitter and receiver side are $\mathbf{n}_{T}$ and $\mathbf{n}_{R}$. Typically, it holds from \cite[Eqs. (7) and (8)]{c:matthaiouhimimo2014} that $\mathbf{n}_{T}\overset{\text{d}}=\mathcal{CN}(\mathbf{0},p \kappa^{2}_{T}\mathbf{I}_{M})$ and $\mathbf{n}_{R}\overset{\text{d}}=\mathcal{CN}(\mathbf{0},p \kappa^{2}_{R}M\mathbf{I}_{N})$, whereas $\kappa_{T}$ and $\kappa_{R}$ represent certain parameters leveraging the residual hardware impairments at the transmitter and receiver. The \emph{Gaussianity} of these distributions has been verified experimentally \cite[Fig. 4.13]{b:wenk2010mimo} where it relies on the distortion noise describing the aggregate effect of several residual hardware impairments \cite{Bjornsonhi2014,j:zhang2015achievable}\footnote{In this current study, it is assumed that compensation algorithms are applied, which mitigate the main hardware impairments \cite{schenk2008rf}.}. Note that $\kappa_{T}$ and $\kappa_{R}$ are inherently associated with the error-vector magnitude (EVM) metric \cite{b:holma2011lte}, which is widely used to quantify the mismatch between the intended and the actual signal in RF transceivers. In addition, EVM facilitates the identification of specific types of degradations encountered in typical wireless transceivers, such as the I/Q phase imbalance, local oscillator phase noise, carrier leakage, nonlinearity, and local oscillator frequency error \cite{j:liu2006evm,j:Georgiadis2004evm}. EVM is defined as
\begin{align}
\text{EVM}\triangleq \sqrt{\frac{\mathbb{E}[\left\|\mathbf{n}_{l}\right\|^{2}]}{\mathbb{E}[\left\|\mathbf{b}\right\|^{2}]}}=\kappa_{l},
\end{align}
with $l\in\{T,R\}$ and $\mathbf{b}\in\{\mathbf{s},\hat{\mathbf{y}}\}$, correspondingly. The architecture of the considered system is illustrated in Fig. \ref{fig1a}.
\begin{figure}[!t]
\centering
\includegraphics[trim=0.0cm 0.0cm 0.0cm 0.0cm, clip=true,totalheight=0.09\textheight]{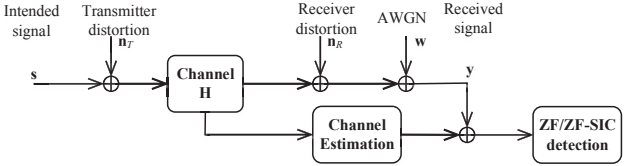}
\caption{Block diagram of the considered system. Note that the transmit signal $\mathbf{s}$ consists of $M$ independent streams with non-identical statistics, while the received signal $\mathbf{y}$ consists of $N$ coefficients (i.e., $N$ receive antennas) where it generally holds that $N\geq M$.}
\label{fig1a}
\end{figure}

In the presence of perfect channel estimation, i.e., when $\hat{\mathbf{H}}=\mathbf{H}$, (\ref{eq11}) coincides with (\ref{eq1}). In the more realistic scenario of channel estimation errors, we have
\begin{align}
\hat{\mathbf{H}}\triangleq \mathbf{H}+\sigma \mathbf{\Omega},
\label{esthdef}
\end{align}
where $\mathbf{\Omega}\overset{\text{d}}= \mathcal{CN}(\mathbf{0},\mathbf{I}_{N})$ whereas $\mathbf{\Omega}\in \mathbb{C}^{N\times M}$ is the channel estimation error matrix, which is uncorrelated with the true channel matrix $\mathbf{H}$.\footnote{This statistical feature can be captured by adopting either maximum-likelihood or MMSE channel estimation \cite{j:WangMurch2007,j:Narasimhan2005}.} Also, $0\leq \sigma \leq 1$ is a measure of the channel estimation accuracy. Finally, the (true) Rician channel fading matrix with non-identical statistics is defined as
\begin{align}
\mathbf{H}\triangleq \mathbf{H}_{d}+\mathbf{H}_{r},
\label{ricianchanndef}
\end{align}
where
\begin{align}
\mathbf{H}_{d}\triangleq [\mathbf{h}_{d}\ \underbrace{\mathbf{0}_{(N\times 1)}\ \cdots \mathbf{0}_{(N\times 1)}}_{(M-1)\text{ column vectors}}],
\label{hddef}
\end{align}
and
\begin{align}
\mathbf{H}_{r}\overset{\text{d}}=\mathcal{CN}\left(\mathbf{0},\mathbf{I}_{N}\otimes \diag\left\{\frac{d^{-a_{i}}_{i}}{(K+1)}\right\}^{M}_{i=1}\right),
\label{hrdef}
\end{align}
denoting the deterministic and random components of $\mathbf{H}$, respectively. Also, $\mathbf{0}_{(N\times 1)}$ stands for an $N$-sized vector with all its coefficients set to zero, $\alpha_{i} \in \{2,6\}$ represents the path-loss factor corresponding to propagation scenarios from free-space path loss to dense urban path loss \cite[Table 2.2]{b:goldsmith}, $K$ denotes the Rician-$K$ factor, and $d_{i}$ is the normalized  (with respect to $1$km) link-distance between the $i$th transmitter and the receiver. For ease of exposition and without loss of generality, the Rician-faded signal is assumed the left-most one with respect to $\mathbf{H}$. Therefore, it holds that
\begin{align}
[\mathbf{H}_{d}]_{q,1}\triangleq \mathbf{h}_{d}=\sqrt{\frac{d^{-a_{1}}_{1}K}{(K+1)}}\exp\left(-\frac{\mathbbmtt{j}(q-1)2\pi D\sin(\varphi)}{\lambda}\right),
\label{hddefff}
\end{align}
where $\mathbbmtt{j}\triangleq \sqrt{-1}$ and $1\leq q\leq N$. Moreover, $\varphi$ is the arrival angle of the LOS signal from the corresponding transmitter measured relative to the array boresight, $D=\lambda/2$ is the antenna spacing at the receiver, while $\lambda$ is the signal transmission wavelength. Thus, it stems that \cite{j:Zhang2014}
\begin{align}
\mathbf{H}\overset{\text{d}}=\mathcal{CN}\left(\mathbf{H}_{d},\mathbf{I}_{N}\otimes \diag\left\{\frac{d^{-a_{i}}_{i}}{(K+1)}\right\}^{M}_{i=1}\right).
\end{align}

\section{Statistics of SNR}
\label{Statistics of SNR}
In the case of the classical ZF detector, upon the signal reception, the ZF filter is applied at the receiver, i.e., the Moore-Penrose pseudoinverse operation at the estimated channel matrix. It is defined as $\hat{\mathbf{H}}^{\dagger}\triangleq (\hat{\mathbf{H}}^{\mathcal{H}}\hat{\mathbf{H}})^{-1}\hat{\mathbf{H}}^{\mathcal{H}}$ yielding that
\begin{align}
\hat{\mathbf{H}}^{\dagger}\mathbf{y}\triangleq \mathbf{r},
\label{rzf}
\end{align}
where $\mathbf{r}$ represents the estimated symbol vector.

\begin{ddef}
In the case of the more sophisticated ZF-SIC detector, (\ref{rzf}) is repeated recursively in $M$ consecutive stages, by replacing $M$ with $M-i+1$ at the $i$th stage. This denotes the corresponding interference nulling and symbol removal from the remaining signal at the corresponding SIC stage. This paper focuses on the ZF-SIC scenario, whereas the conventional ZF scenario arises as a special case by setting $i=1$ (i.e., only one stage with concurrent detection of $M$ symbols).
\end{ddef}

Typically, it holds that $\sigma \ll 1$ \cite{j:WangMurch2007,j:Narasimhan2005}, thereby, $\hat{\mathbf{H}}^{\dagger}$ can be sufficiently approximated by the linear part of its Taylor expansion as
\begin{align}
\hat{\mathbf{H}}^{\dagger}\approx \mathbf{H}^{\dagger}\left(\mathbf{I}_{N}-\sigma \mathbf{\Omega}\mathbf{H}^{\dagger}\right).
\label{taylor}
\end{align}
Note that when $\sigma=0$ (i.e., in the case of perfect channel estimation), (\ref{taylor}) becomes exact since $\hat{\mathbf{H}}^{\dagger}=\mathbf{H}^{\dagger}$.

Using (\ref{taylor}) into the left-hand side of (\ref{rzf}), we have that
\begin{align}
\nonumber
\mathbf{r}&\approx \mathbf{H}^{\dagger}\left(\mathbf{I}_{N}-\sigma \mathbf{\Omega}\mathbf{H}^{\dagger}\right)\mathbf{y}\\
\nonumber
&=\mathbf{H}^{\dagger}\left(\mathbf{I}_{N}-\sigma \mathbf{\Omega}\mathbf{H}^{\dagger}\right)\left(\mathbf{H} (\mathbf{s}+\mathbf{n}_{T})+\mathbf{n}_{R}+\mathbf{w}\right)\\
\nonumber
&=\mathbf{s}+\mathbf{n}_{T}+\mathbf{H}^{\dagger}(\mathbf{n}_{R}+\mathbf{w})-\sigma \mathbf{H}^{\dagger} \mathbf{\Omega} \mathbf{s}\\
\nonumber
&\ \ \ -\sigma \mathbf{H}^{\dagger} \mathbf{\Omega} \mathbf{n}_{T}-\sigma \mathbf{H}^{\dagger} \mathbf{\Omega} \mathbf{H}^{\dagger} (\mathbf{n}_{R}+\mathbf{w})\\
&=\mathbf{s}+\mathbf{w}',
\label{rrzf}
\end{align}
where
\begin{align}
\nonumber
\mathbf{w}'&\triangleq \mathbf{n}_{T}+\mathbf{H}^{\dagger}(\mathbf{n}_{R}+\mathbf{w})-\sigma \mathbf{H}^{\dagger} \mathbf{\Omega} (\mathbf{s}+\mathbf{n}_{T})\\
&-\sigma \mathbf{H}^{\dagger} \mathbf{\Omega} \mathbf{H}^{\dagger} (\mathbf{n}_{R}+\mathbf{w}).
\label{w'}
\end{align}

\begin{cor}
The instantaneous approximate received SNR for the $i$th stream, $1\leq i \leq M$, is expressed as
\begin{align}
\text{SNR}_{i}&=\frac{p}{\mathbb{E}[\mathbf{w}'\mathbf{w}'^{\mathcal{H}}]_{ii}}=\frac{1}{\kappa^{2}_{T}+\psi \left[(\mathbf{H}^{\mathcal{H}}\mathbf{H})^{-1}\right]_{ii}}
\label{snr}
\end{align}
where
\begin{align}
\nonumber
\psi&=\bigg(\kappa^{2}_{R}M+\frac{N_{0}}{p}+\sigma^{2}M(1+\kappa^{2}_{T})+\sigma^{2}\left(\kappa^{2}_{R}M+\frac{N_{0}}{p}\right)\\
&\times \Tr\left[(\mathbf{H}^{\mathcal{H}}\mathbf{H})^{-1}\right]\bigg).
\label{psi}
\end{align}
\end{cor}

\begin{proof}
The proof is presented in Appendix \ref{appsnr}.
\end{proof}

The following lemmas are key results for the subsequent analysis.

{\color{black}\begin{lem}
In the case of i.n.i.d. rank-$1$ Rician fading channels, it holds that
\begin{align}
\mathcal{Y}_{i}\triangleq \frac{1}{[(\mathbf{H}^{\mathcal{H}}\mathbf{H})^{-1}]_{ii}}\overset{\text{d}}=\left(\frac{d^{-\alpha_{i}}_{i}}{(K+1)}\right)\chi^{2}_{N-M+i}(\theta_{i}),\ 1\leq i\leq M,
\label{pdfyy}
\end{align}
where $\theta_{1}\overset{\text{d}}=\text{Beta}(N-M+i,M-1)$, whereas $\theta_{1}$ and $\mathcal{Y}_{1}$ are mutually independent RVs, and $\{\theta_{i}\}^{M}_{i=2}\triangleq 0$. The corresponding PDF is given by \cite[Eq. (29.4)]{b:johnson1994continuous}
\begin{align}
\nonumber
&f_{\mathcal{Y}_{i}|\theta_{i}}(x)=\frac{(K+1)}{2 d^{-\alpha_{i}}_{i}}\exp\left(-\frac{(K+1)x}{d^{-\alpha_{i}}_{i}}-\frac{\theta_{i}}{2}\right)\\
&\times \left(\frac{(K+1)x}{d^{-\alpha_{i}}_{i}\theta_{i}}\right)^{\frac{N-M+i-1}{2}} I_{N-M+i-1}\left(\sqrt{\frac{\theta_{i}(K+1)x}{d^{-\alpha_{i}}_{i}}}\right).
\label{pdfy}
\end{align}
\end{lem}

\begin{proof}
The proof is relegated in Appendix \ref{appncxpdf}.
\end{proof}

It is noteworthy that (\ref{pdfyy}) represents an extension of ZF and ZF-SIC reception under Rayleigh-only fading channels \cite{j:GOREHeathPaulraj,j:Loyka2006}, where the corresponding SNR of each stream follows a central chi-squared PDF, such as $\mathcal{Y}_{i}\overset{\text{d}}=\chi^{2}_{N-M+i}$. For rank-$1$ Rician fading channels, the Rician-faded stream follows a non-central chi-squared PDF, conditioned on its non-centrality parameter. In turn, SNR of the Rayleigh-faded $j$th stream, with $2\leq j \leq M$, follows an unconditional central chi-squared PDF, by setting $\theta_{i}=0$ in (\ref{pdfy}), which is influenced from the presence of the Rician-faded signal by a ($K+1$)-scaling factor \cite{j:siriteanuexact2014}. 

\begin{lem}
The unconditional CDF $F_{\mathcal{Y}_{i}}(\cdot)$ can be derived as
\begin{align}
\nonumber
&F_{\mathcal{Y}_{1}}(x)=1-\frac{1}{B(N-M+1,M-1)}\sum^{M-2}_{j=0}\frac{\binom{M-2}{j}}{(N-M+j+1)}\\
\nonumber
&\times \vast\{Q_{N-M+1}\left(1,\sqrt{\frac{(K+1)x}{d^{-\alpha_{1}}_{1}}}\right)-\left(\sqrt{\frac{(K+1)x}{d^{-\alpha_{1}}_{1}}}\right)^{N-M+1}\\
\nonumber
&\times \vast[Q_{2(N-M+j+1),N-M+1}\left(\sqrt{\frac{(K+1)x}{d^{-\alpha_{1}}_{1}}},0\right)\\
&-Q_{2(N-M+j+1),N-M+1}\left(\sqrt{\frac{(K+1)x}{d^{-\alpha_{1}}_{1}}},1\right)\vast]\vast\},
\label{cdfyyyyyy}
\end{align}
and
\begin{align}
\nonumber
&F_{\mathcal{Y}_{i}}(x)=\\
\nonumber
&1-Q_{N-M+i}\left(\sqrt{\theta_{i}},\sqrt{\frac{(K+1)x}{d^{-\alpha_{i}}_{i}}}\right)\\
&=1-\frac{\Gamma\left(N-M+i,\frac{(K+1)x}{2d^{-\alpha_{i}}_{i}}\right)}{\Gamma(N-M+i)},\ 2\leq i\leq M,\ \{\theta_{i}=0\}^{M}_{i=2}.
\label{cdfyyyy}
\end{align}
\end{lem}

\begin{proof}
Please refer to Appendix \ref{appcdfclosform} for the detailed proof.
\end{proof}

The CDF $F_{\mathcal{Y}_{1}}(\cdot)$ for the Rician-faded $1$st stream, as given in (\ref{cdfyyyyyy}), is provided in a closed form in terms of finite sum series of the Marcum-$Q$ and Nuttall-$Q$ functions. Unfortunately, the Nuttall-$Q$ function with arbitrary parameters is not included as a standard build-in function in most popular mathematical software platforms. Nevertheless, the certain Nuttall-$Q$ functions in (\ref{cdfyyyyyy}) admit alternative closed-form solutions, which are presented, respectively, as \cite[Eqs. (6) and (7)]{c:Sofotasiosnuttall}
\begin{align}
\nonumber
&Q_{2(N-M+j+1),N-M+1}\left(\sqrt{\frac{(K+1)x}{d^{-\alpha_{1}}_{1}}},0\right)=\\
\nonumber
&\frac{\Gamma\left(\frac{3(N-M+1)+2j+1}{2}\right)\left(\frac{(K+1)x}{d^{-\alpha_{1}}_{1}}\right)^{\frac{N-M+1}{2}}}{2^{\frac{M-N-2j-2}{2}}\Gamma(N-M+2)\exp\left(\frac{(K+1)x}{2d^{-\alpha_{1}}_{1}}\right)}\\
&\times \textstyle {}_1F_{1}\left(\frac{3(N-M+1)+2j+1}{2},N-M+2;\frac{(K+1)x}{2d^{-\alpha_{1}}_{1}}\right),
\label{nuttall1}
\end{align} 
and\footnote{The expression in (\ref{nuttall2}) is valid only when $N$ and $M$ are both even numbers, such that the parameter $N-M+1$ becomes odd. Indeed, this assumption meets practical considerations since most antenna array architectures rely on an even number of antennas. This occurs to accommodate hybrid couplers and power dividers/splitters or to formulate input/output ports of selective matrices (e.g., Butler matrices) for beamforming, e.g., see \cite[\S 6]{b:Balanis}.}
\begin{align}
\nonumber
&Q_{2(N-M+j+1),N-M+1}\left(\sqrt{\frac{(K+1)x}{d^{-\alpha_{1}}_{1}}},1\right)=\\
\nonumber
&\sum^{j+\frac{N-M}{2}+1}_{l=1}\frac{2^{j+\frac{N-M}{2}-l+1}(j+\frac{N-M}{2})!\binom{j+\frac{5(N-M)}{2}}{j+\frac{N-M}{2}-l+1}}{(l-1)!}\\
\nonumber
&\times \left(\frac{(K+1)x}{d^{-\alpha_{1}}_{1}}\right)^{\frac{N-M+2 l-1}{2}}Q_{N-M+l+1}\left(\sqrt{\frac{(K+1)x}{d^{-\alpha_{1}}_{1}}},1\right)\\
\nonumber
&+\exp\left(-\frac{\left(\frac{(K+1)x}{d^{-\alpha_{1}}_{1}}+1\right)}{2}\right)\sum^{j+\frac{N-M}{2}}_{l=1}\sum^{j+\frac{N-M}{2}-l}_{s=0}\\
&\times \frac{2^{j+\frac{N-M}{2}-l-s}(j+\frac{N-M}{2}-s-1)!\binom{j+\frac{5(N-M)}{2}}{j+\frac{N-M}{2}-l-s}}{(l-1)!\left(\frac{(K+1)x}{d^{-\alpha_{1}}_{1}}\right)^{\frac{1-l}{2}}}\\
&\times I_{N-M+l}\left(\sqrt{\frac{(K+1)x}{d^{-\alpha_{1}}_{1}}}\right).
\label{nuttall2}
\end{align} 

\begin{lem}
The CDF of SNR for the $i$th stream in (\ref{snr}) is approached by
\begin{align}
F_{\text{SNR}_{i}}(x)\approx F_{\mathcal{Y}_{i}}\left(\frac{\left(\kappa^{2}_{R}M+\frac{N_{0}}{p}+\sigma^{2}M(1+\kappa^{2}_{T})\right) x}{1-\kappa^{2}_{T}x}\right).
\label{cdfy}
\end{align}
\end{lem}

\begin{proof}
The proof is provided in Appendix \ref{appcdfy}.
\end{proof}

Collecting the aforementioned statistical results, we summarize the following insightful observations:

\begin{rem}
The computation of (\ref{cdfyyyyyy}) is presented in a closed-form solution in terms of finite sum series of the Marcum-$Q$ function, Kummer's confluent hypergeometric function and modified Bessel function of the first kind, which all are included as standard build-in functions in most popular mathematical software platforms. Similar non closed-form representations of (\ref{cdfyyyyyy}) in terms of infinite series including special functions have been reported in \cite{j:Siriteanurank0ne,siriteanu2016chi,j:siriteanuexact2014,j:SiriteanuMiyanaga2012,j:Siriteanurician2014}. On the other hand, the derived result in (\ref{cdfyyyyyy}) is provided in an exact closed form, whereas it is accurate and computationally efficient.
\label{remm1}
\end{rem}
}

\begin{rem}
The CDF of SNR for each received stream given in (\ref{cdfy}) is exact in the case when hardware impairments occur at both the transmitter and receiver sides under perfect channel estimation. When imperfect channel estimation conditions occur, (\ref{cdfy}) approximates the actual SNR for each stream (when ZF or ZF-SIC is applied at the receiver). The approximation error is marginal for reasonable ranges of channel estimation imperfection and/or signal distortion due to hardware impairments (i.e., when $\{\kappa^{2}_{T},\kappa^{2}_{R},\sigma^{2}\}\ll 1$). This approximation error vanishes for $N\rightarrow \infty$, i.e., in massive MIMO antenna systems.
\label{rem1}
\end{rem}

At this point, it is noteworthy that mmWave transmission results to higher Doppler spreads for a given user velocity. Nevertheless, mmWave communication systems are expected to operate on a relatively short range (e.g., femto- and/or pico-cell coverage areas) due to their increased path loss. Consequently, this reflects to a rather low user mobility and thus a degrease of the corresponding velocity by an order of magnitude as compared to conventional systems, operating near $2-3$ GHz. In addition, the specific scenario of mmWave transmission (i.e., LOS or near-LOS propagation) under massive MIMO infrastructures results to an increase in coherence bandwidth with the aid of the so-called `\emph{pencil-beam}' communication \cite{j:LuSwindlehurst2014}. In such an environment, these communication systems do not require a significant increase in channel update rates, thereby, resulting to marginal errors in channel estimation \cite[\S VIII.B.2]{j:LuSwindlehurst2014}.

\section{Performance Metrics}
\label{Performance Metrics}
The outage probability and ergodic capacity for each stream are analytically presented into this section.

\subsection{Outage Probability}
Outage probability of the $i$th stream ($1\leq i\leq M$), $P^{(i)}_{\text{out}}(\gamma_{\text{th}})$, is defined as the probability that the SNR of the $i$th stream falls below a certain threshold value $\gamma_{\text{th}}\triangleq 2^{\mathcal{R}}-1$, where $\mathcal{R}$ stands for a given data transmission rate in bps/Hz.

{\color{black}\begin{prop}
Outage probability of the Rician-faded transmitted stream (i.e., $i=1$) is directly obtained from (\ref{cdfy}), (\ref{nuttall1}), (\ref{nuttall2}) and (\ref{cdfyyyyyy}) as
\begin{align}
\nonumber
&P^{(1)}_{\text{out}}(\gamma_{\text{th}})=1-\frac{1}{B(N-M+1,M-1)}\sum^{M-2}_{j=0}\frac{\binom{M-2}{j}}{(N-M+j+1)}\\
\nonumber
&\times \vast\{Q_{N-M+1}\left(1,\sqrt{\frac{\left(\frac{\left(\kappa^{2}_{R}M+\frac{N_{0}}{p}+\sigma^{2}M(1+\kappa^{2}_{T})\right) \gamma_{\text{th}}}{1-\kappa^{2}_{T}\gamma_{\text{th}}}\right)}{(K+1)^{-1}d^{-\alpha_{i}}_{i}}}\right)\\
\nonumber
&-\left(\frac{\left(\frac{\left(\kappa^{2}_{R}M+\frac{N_{0}}{p}+\sigma^{2}M(1+\kappa^{2}_{T})\right) \gamma_{\text{th}}}{1-\kappa^{2}_{T}\gamma_{\text{th}}}\right)}{(K+1)^{-1}d^{-\alpha_{i}}_{i}}\right)^{\frac{N-M+1}{2}}\\
\nonumber
&\times \vast[Q_{2(N-M+j+1),N-M+1}\left(\sqrt{\frac{\left(\frac{\left(\kappa^{2}_{R}M+\frac{N_{0}}{p}+\sigma^{2}M(1+\kappa^{2}_{T})\right) \gamma_{\text{th}}}{1-\kappa^{2}_{T}\gamma_{\text{th}}}\right)}{(K+1)^{-1}d^{-\alpha_{i}}_{i}}},0\right)\\
&-Q_{2(N-M+j+1),N-M+1}\left(\sqrt{\frac{\left(\frac{\left(\kappa^{2}_{R}M+\frac{N_{0}}{p}+\sigma^{2}M(1+\kappa^{2}_{T})\right) \gamma_{\text{th}}}{1-\kappa^{2}_{T}\gamma_{\text{th}}}\right)}{(K+1)^{-1}d^{-\alpha_{i}}_{i}}},1\right)\vast]\vast\}.
\label{outrice}
\end{align}
Moreover, outage probability of the Rayleigh-faded $i$th transmitted stream ($2\leq i\leq M$) is directly obtained from (\ref{cdfy}) and (\ref{cdfyyyy}) as
\begin{align}
P^{(i)}_{\text{out}}(\gamma_{\text{th}})=1-\frac{\Gamma\left(N-M+i,\frac{\left(\kappa^{2}_{R}M+\frac{N_{0}}{p}+\sigma^{2}M(1+\kappa^{2}_{T})\right) \gamma_{\text{th}}}{2(K+1)^{-1}d^{-\alpha_{i}}_{i}(1-\kappa^{2}_{T}\gamma_{\text{th}})}\right)}{\Gamma(N-M+i)}.
\label{out}
\end{align}
\end{prop}
Notice that $\gamma_{\text{th}}<1/\kappa^{2}_{T}$ should hold in (\ref{outrice}) and (\ref{out}), which is usually the case\footnote{As an illustrative example, typical values of $\kappa^{2}_{T}$ and $\kappa^{2}_{R}$ in long-term evolution (LTE) infrastructures \cite[Section 14.3.4]{b:holma2011lte} are in the range of $[0.0064-0.0306]$.} in practical communication systems; otherwise, $P^{(i)}_{\text{out}}(\gamma_{\text{th}})=1$.
}

\subsection{Ergodic Capacity}
The ergodic capacity of the $i$th stream is defined as the statistical mean of the instantaneous mutual information between the corresponding transmitter and receiver (in bps/Hz). It is explicitly defined as
\begin{align}
\nonumber
\overline{C}_{i}&\triangleq \mathbb{E}[\log_{2}(1+\text{SNR}_{i})]\\
&=\frac{1}{\text{ln(2)}}\int^{1/\kappa^{2}_{T}}_{0}\frac{(1-F_{\text{SNR}_{i}}(x))}{1+x}dx.
\label{cdef}
\end{align}
Unfortunatelly, (\ref{cdef}) does not have an analytical solution for the general case and can only be calculated via numerical integration. The particular case with a non-impaired hardware at the transmitter side (i.e., when $\kappa_{T}=0$) is analytically presented as follows.

{\color{black}\begin{prop}
The normalized\footnote{For ease of presentation, we normalize the ergodic capacity with the factor $1/\text{ln}(2)\approx 1.442695$.} ergodic capacity of the $1$st (Rician-faded) transmitted stream is given by
\begin{align}
\nonumber
\overline{C}_{1}&=\sum^{\infty}_{k=0}\sum^{N-M+k}_{j=0}\left(\frac{\left(\kappa^{2}_{R}M+\frac{N_{0}}{p}+\sigma^{2}M\right)}{2(K+1)^{-1}d^{-\alpha_{1}}_{1}}\right)^{j}\\
\nonumber
&\times \frac{\exp\left(\frac{\left(\kappa^{2}_{R}M+\frac{N_{0}}{p}+\sigma^{2}M\right)}{2(K+1)^{-1}d^{-\alpha_{1}}_{1}}\right)\Gamma\left(-j,\frac{\left(\kappa^{2}_{R}M+\frac{N_{0}}{p}+\sigma^{2}M\right)}{2(K+1)^{-1}d^{-\alpha_{1}}_{1}}\right)}{k!2^{k}B(N-M+1,M-1)}\\
&\times \sum^{M-2}_{r=0}\frac{\binom{M-2}{r}\gamma\left(N-M+k+r+1,\frac{1}{2}\right)}{2^{M-N-k-r-1}}.
\label{ergcap}
\end{align}
As indicated in the next section, (\ref{ergcap}) is a rapidly converging series for reasonable (i.e. practical) accuracy levels of the achievable data rate, becoming an efficient approach. Also, the corresponding normalized ergodic capacity of the $i$th Rayleigh-faded transmitted stream ($2\leq i\leq M$) is given by
\begin{align}
\nonumber
\overline{C}_{i}&=\sum^{N-M+i-1}_{j=0}\frac{\left(\frac{\left(\kappa^{2}_{R}M+\frac{N_{0}}{p}+\sigma^{2}M\right)}{2(K+1)^{-1}d^{-\alpha_{i}}_{i}}\right)^{j}}{\exp\left(-\frac{\left(\kappa^{2}_{R}M+\frac{N_{0}}{p}+\sigma^{2}M\right)}{2(K+1)^{-1}d^{-\alpha_{i}}_{i}}\right)}\\
&\times \Gamma\left(-j,\frac{\left(\kappa^{2}_{R}M+\frac{N_{0}}{p}+\sigma^{2}M\right)}{2(K+1)^{-1}d^{-\alpha_{i}}_{i}}\right).
\label{ergcapray}
\end{align}
\end{prop}

\begin{proof}
The proof is presented in Appendix \ref{appergcap}.
\end{proof}
}

Furthermore, it trivially follows that the sum-capacity of the entire system is defined as
\begin{align}
\overline{C}_{\sum}\triangleq \sum^{M}_{i=1}\overline{C}_{i},
\label{ergcapsum}
\end{align}
which is reduced to $\overline{C}_{\sum}=\overline{C}_{1}+(M-1)\overline{C}_{2}$ (due to the symmetry amongst the Rayleigh-faded received streams), in the case of the conventional ZF (non-SIC) reception.

From an engineering standpoint, the following remarks summarize some useful outcomes:

\begin{rem}
The outage probability and ergodic capacity of the $i$th stream for an ideal ZF-SIC communication system (i.e., with perfect channel estimates and no hardware impairments) are directly obtained from (\ref{outrice}), (\ref{out}), (\ref{ergcap}) and (\ref{ergcapray}), respectively, by setting $\{\sigma,\kappa_{T},\kappa_{R}\}=0$.
\label{rem2}
\end{rem}

\begin{rem}
In the high SNR regime, when $p/N_{0}\rightarrow +\infty$, the $N_{0}/p$ term can be neglected from (\ref{outrice}) and (\ref{out}), denoting a \emph{non-zero} outage floor. This is in contrast to the ideal scenario (i.e., $\{\sigma,\kappa_{T},\kappa_{R}\}=0$), where the outage probability of each stream is vanished (asymptotically tend to zero). Similarly, using the mentioned argument in (\ref{ergcap}) and (\ref{ergcapray}), an upper bound on the achievable data rate for each stream is obtained.
\label{rem3}
\end{rem}

\begin{rem}
By using the standard properties of the Marcum-$Q$ and Nuttall-$Q$ functions \cite[Eq. (2.10)]{j:GilAmparo2014} and \cite[Th. 3]{j:Kapinasnuttall}, we can observe that the outage probability of the Rician-faded stream is a decreasing function with respect to $N-M+i$ (i.e., for less transmitters with a fixed number of receive antennas and/or for more receive antennas for a given number of transmitters) and the Rician-$K$ factor. In addition, it is an increasing function with respect to the system communication imperfections (i.e., channel estimation errors and hardware impairments), as expected. The two parameters of both the Marcum-$Q$ and Nuttall-Q functions in (\ref{outrice}) are both affected by a factor of $1/2$, whereas the corresponding order of these functions are linearly scaled (affected by a unit-factor). Thereby, it turns out that the condition $N\gg M$ plays a more critical role to the system performance than communication imperfections or propagation losses (e.g., variations on the LOS signal power). Similar observations are obtained for the outage probability of the remaining Rayleigh-faded streams by using (\ref{out}) and (\ref{cdfyyyy}).
\label{rem4}
\end{rem}

\section{Numerical Results}
\label{Numerical Results}
In this section, analytical results are presented and compared with Monte Carlo simulations. For ease of tractability and without loss of generality, we assume symmetric levels of impairments at the transceiver, i.e., equal hardware quality at the transmitters and receiver by setting $\kappa_{T}=\kappa_{R}\triangleq \kappa$, $N_{0}=1$, while $\varphi=20^{\circ}$. To gain more insights, we assume a normalized distance for each node denoted as $d=1$ (unless otherwise stated). In addition, the path-loss exponent is set to be $\alpha=4$, corresponding to a classic urban environment \cite[Table 2.2]{b:goldsmith}. The 1st SIC stage corresponds to the detection of Rician-faded stream (see (\ref{hddef})), while all the subsequent ones ($M-1$) correspond to the Rayleigh-faded streams. In all the figures, simulation points and analytical results are denoted by circle-marks and line-curves, respectively. 

It is obvious from Fig. \ref{fig1} that the outage performance reduces for increased Rician-$K$ values. This effect occurs because $\mathbf{H}\rightarrow \mathbf{H}_{d}$ as Rician-$K$ increases (i.e., the LOS component dominate). Both the imperfect channel estimation and hardware impairments dramatically affect the outage probability. The performance of the first SIC stage is illustrated (i.e., when $i=1$), which serves as a lower system performance bound. This is due to the fact that the amount of co-channel interference is reduced (canceled) for each successive stage, and therefore, the corresponding SNR is increased, which in turn shows that the performance of the $j$th stage is enhanced as compared to the $i$th one, where $j>i$ \cite[\S IV]{j:WCLMiridakisVergados2013}.
\begin{figure}[!t]
\centering
\includegraphics[trim=1.5cm 0.3cm 2.5cm 1.2cm, clip=true,totalheight=0.26\textheight]{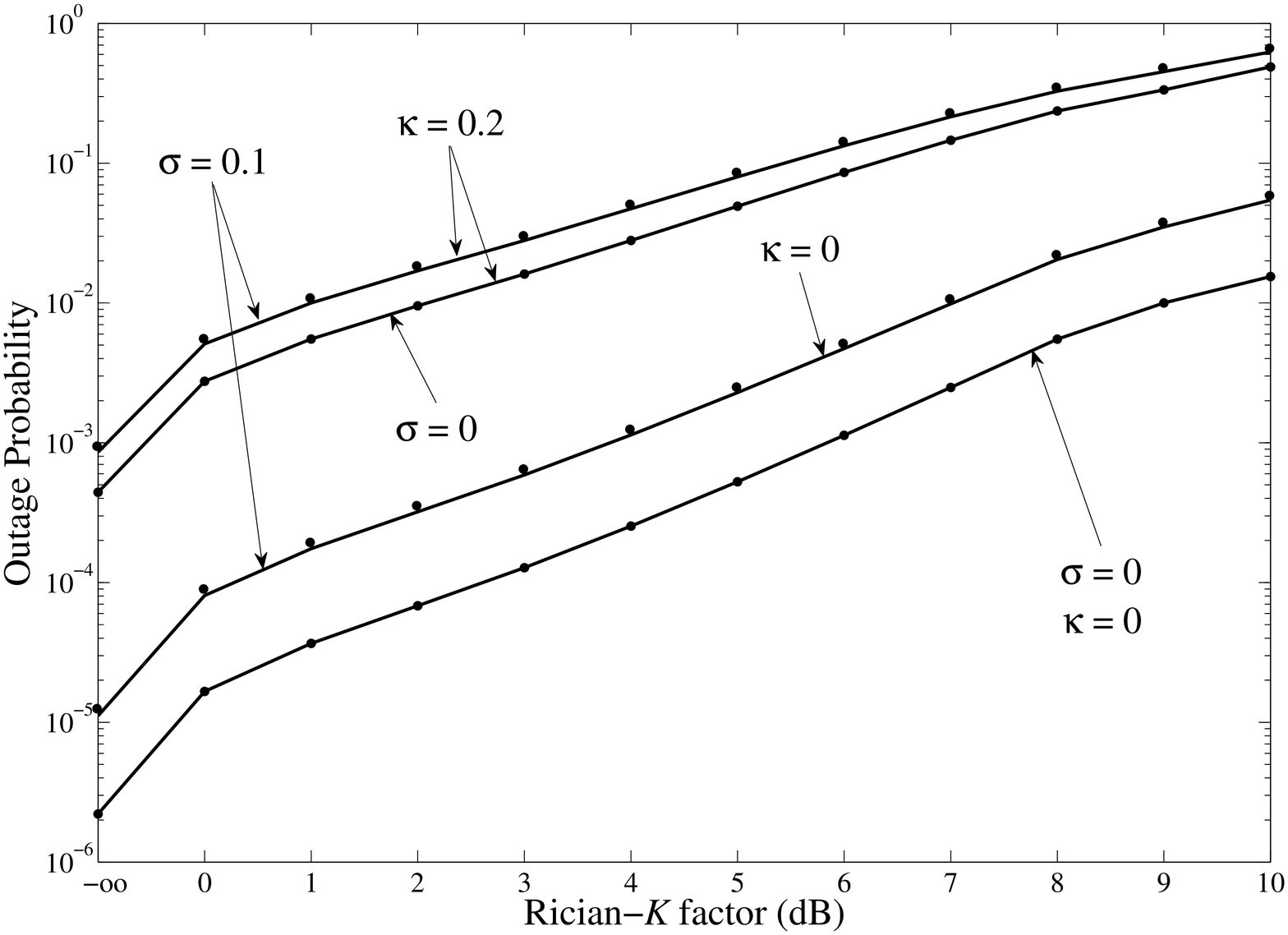}
\caption{Outage probability vs. various values of the Rician-$K$ factor, when $N=8$, $M=4$, $i=1$ (i.e., the 1st SIC stage), $\gamma_{\text{th}}/N_{0}=6$dB, and $p/N_{0}=10$dB.}
\label{fig1}
\end{figure}

Figure \ref{fig2} shows the outage performance for various SNR regions and system imperfections. Notice that only the ideal system setup with no imperfections asymptotically tends to zero. All the other scenarios reach an outage (asymptotic) floor, where no diversity order occurs in these cases. The mentioned outage floor explicitly follows Remark \ref{rem3}; the relevant outage floor curves, however, have not been depicted to avoid clutter. The corresponding array order and outage floor are directly related to the amount of imperfections of the channel estimation and impaired hardware.
\begin{figure}[!t]
\centering
\includegraphics[trim=1.5cm 0.2cm 2.5cm 1.2cm, clip=true,totalheight=0.26\textheight]{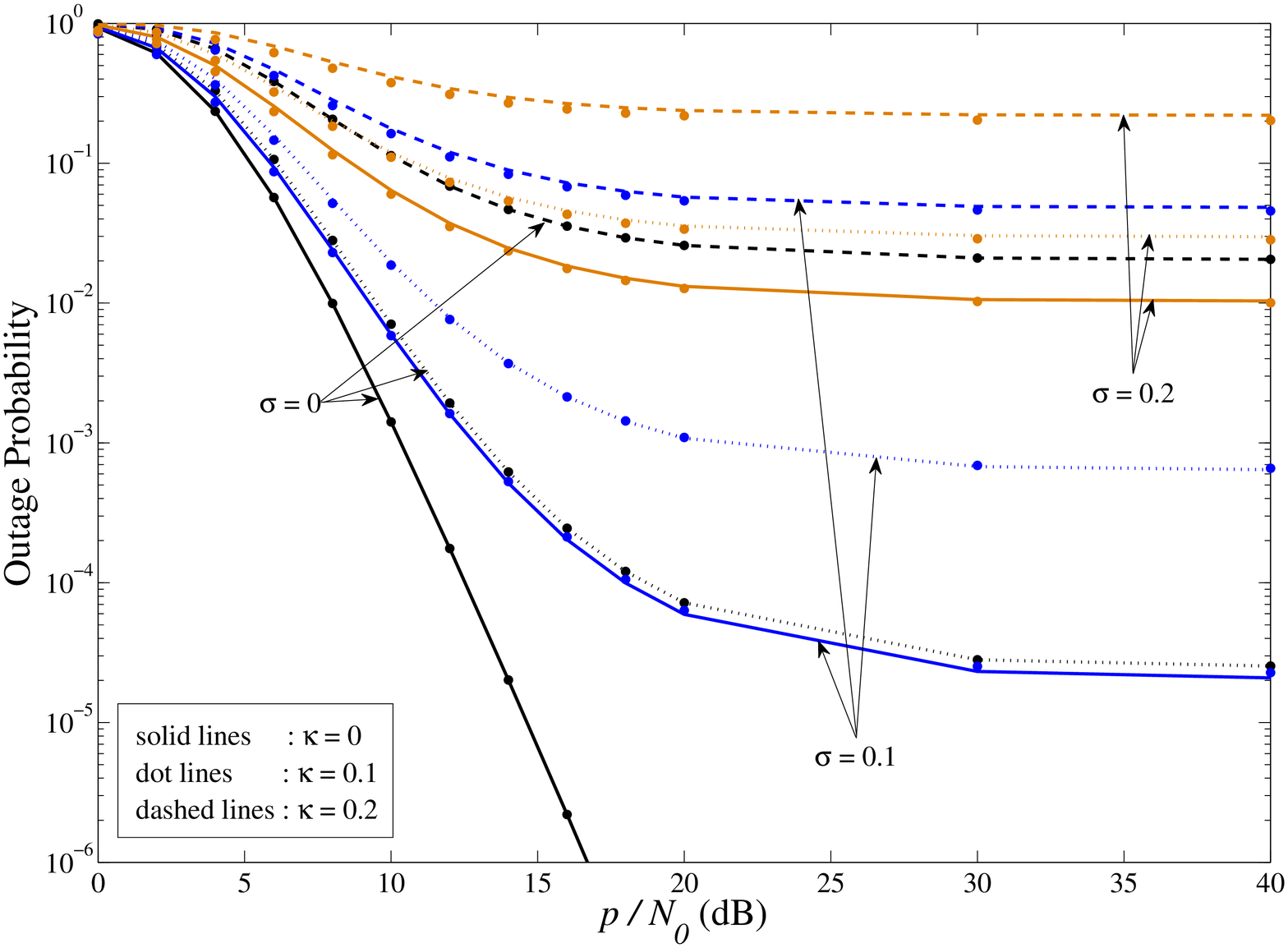}
\caption{Outage probability vs. various values of average SNR, when $N=8$, $M=4$, $i=1$ (i.e., the 1st SIC stage), $\gamma_{\text{th}}/N_{0}=6$dB, and $K=7$dB.}
\label{fig2}
\end{figure}

A beneficial performance improvement arising when adopting the SIC-enabled reception is illustrated in Fig. \ref{fig3}. In particular, the performance of the first SIC stage coincides with the conventional ZF (non-SIC) approach, while the outage performance is enhanced for consecutive SIC stages (as expected). In addition, placing more antennas for reception greatly enhances the overall outage performance, regardless of the amount of system imperfections, which is in agreement with Remark \ref{rem4}.
\begin{figure}[!t]
\centering
\includegraphics[trim=1.5cm 0.2cm 2.5cm 1.2cm, clip=true,totalheight=0.26\textheight]{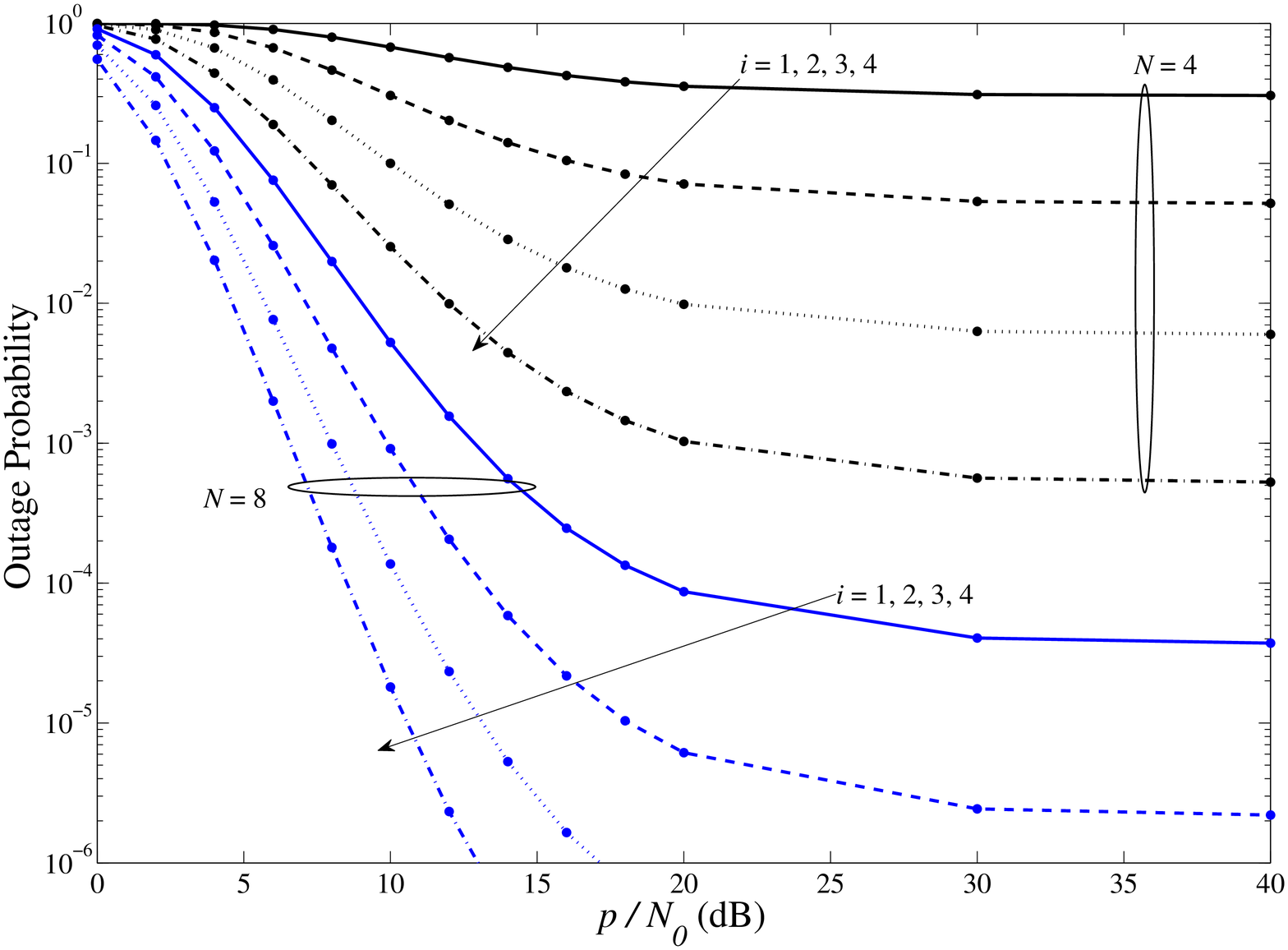}
\caption{Outage probability vs. various values of average SNR, when $M=4$, $\gamma_{\text{th}}/N_{0}=3$dB, and $K=10$dB. Also, $\sigma=0.05$ and $\kappa=0.1$.}
\label{fig3}
\end{figure}

In general, ergodic capacity represents a more solid performance metric than outage probability since it is not application-dependent, i.e., outage performance is directly related to the outage threshold, which can be considered as an application-dependent parameter. Ergodic capacity can be considered as a system-dependent metric since it is related only to the system configuration and networking infrastructure (e.g., number of antennas, modulation scheme, and transmission power). Therefore, the following numerical results use this performance metric.
In addition, outage probability may not be the suitable performance metric for the analysis and evaluation of massive MIMO deployments since it has low values for larger receive antenna arrays. Eventually, outage becomes negligible for practical applications in moderate-to-high SNR regions. It is more convenient and/or preferable to use the average ergodic capacity as an efficient performance tool in massive MIMO systems.

Due to this reason, Fig. \ref{fig6} presents the normalized sum ergodic capacity in very high (yet finite) values of the receive antenna array. An insightful observation is that the presence of a LOS signal affects the overall performance, even when $N\gg M$. As Fig. \ref{fig7} indicates, the performance enhancement produced by a SIC-enabled reception becomes more beneficial when the available $(N-M+i)$ DoF are reduced (e.g., when $N=M$). In summary, the derived analytical result for the ergodic capacity of the Rician-faded signal are based on a fast converging series. An illustrative example is provided in Table \ref{Table}.
\begin{figure}[!t]
\centering
\includegraphics[trim=1.5cm 0.2cm 2.5cm 1.2cm, clip=true,totalheight=0.26\textheight]{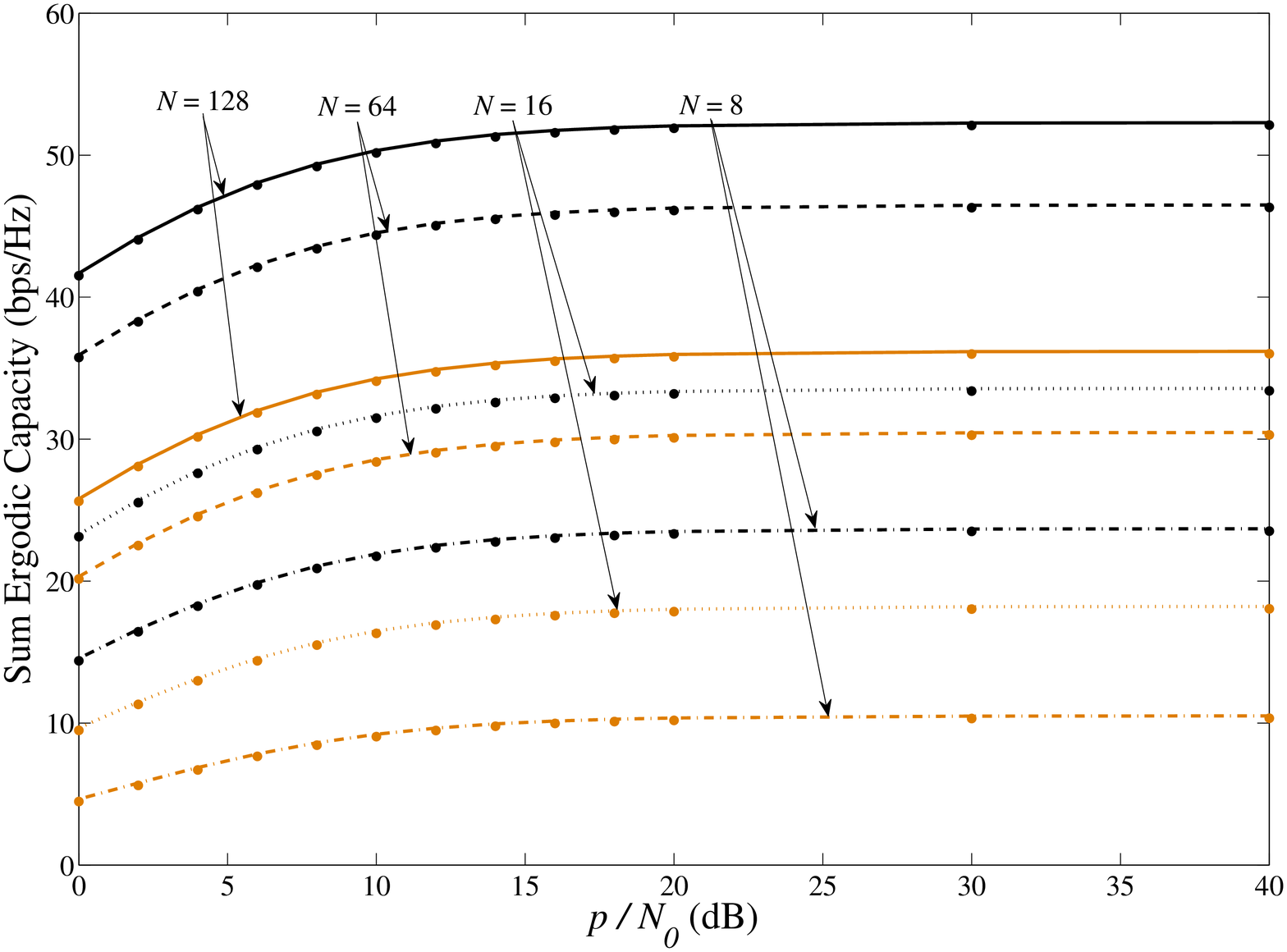}
\caption{(Normalized) Sum ergodic capacity of ZF-SIC vs. various values of average SNR, when $M=8$, and $\{\sigma,\kappa\}=0.15$. Also, black and orange lines correspond to the cases when $K=-\infty$dB (i.e., Rayleigh fading) and $K=10$dB, respectively.}
\label{fig6}
\end{figure}

\begin{figure}[!t]
\centering
\includegraphics[trim=1.5cm 0.2cm 2.5cm 1.2cm, clip=true,totalheight=0.26\textheight]{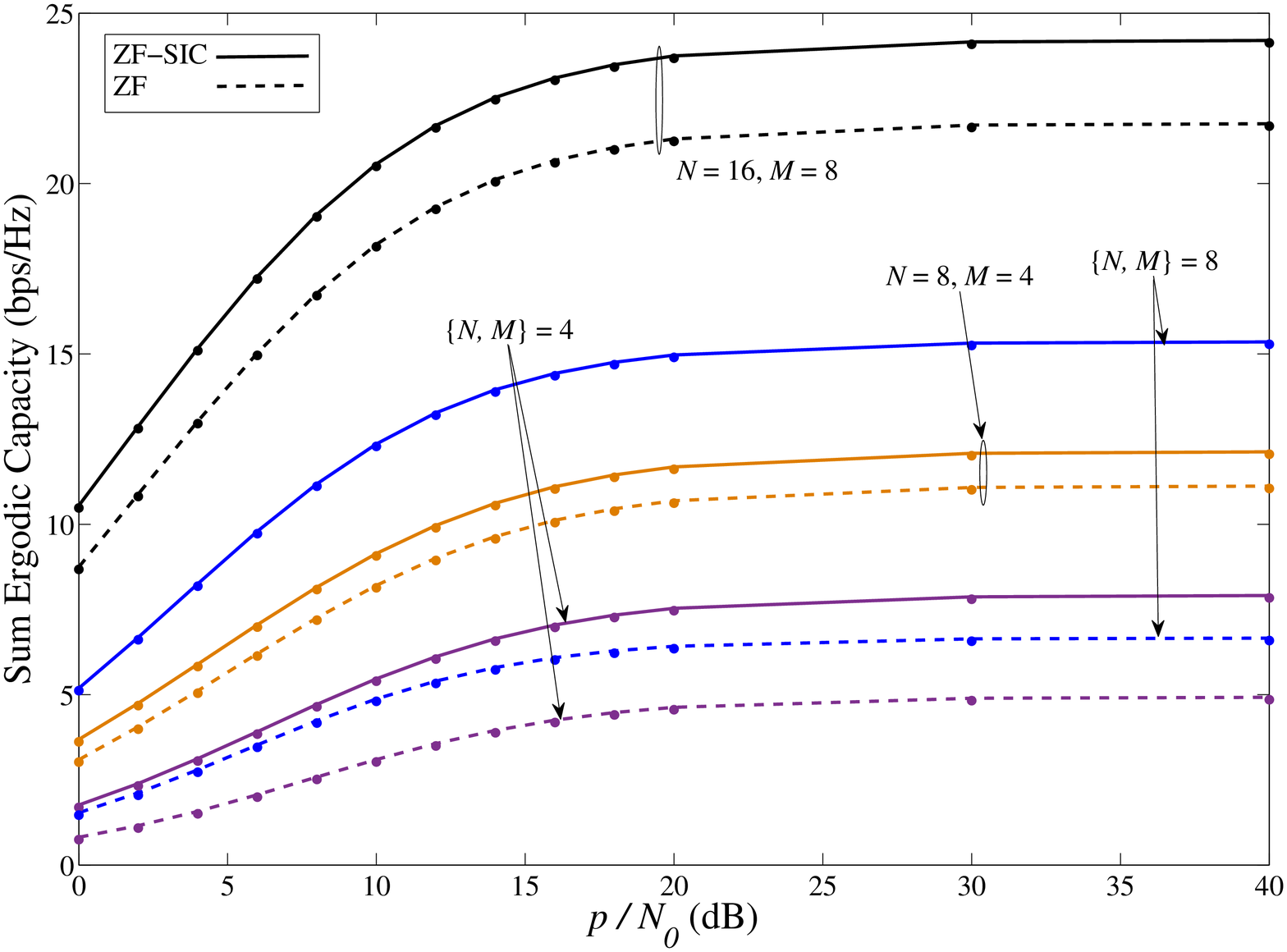}
\caption{(Normalized) Sum ergodic capacity vs. various values of average SNR, when $\{\sigma,\kappa\}=0.1$ and $K=10$dB.}
\label{fig7}
\end{figure}

\begin{table}[!t]
\caption{$T$-TERMS REQUIRED TO BE SUMMED IN (\ref{ergcap}) TO ACHIEVE ACCURACY UP TO THE $3$rd DECIMAL PLACE}
\label{Table}
\begin{center}
\begin{tabular}{l l | l}\hline
$N$ & $M$ & $T$\\\hline
4 & 4 & 3\\
8 & 4 & 9\\
8 & 8 & 4\\
16 & 8 & 13\\
64 & 8 & 42\\
128 & 8 & 76\\\hline
\end{tabular}
\end{center}
*$p/N_{0}=10$dB, $K=6$dB, and $T$ denotes the reserved sum terms. Similar conclusions have been observed for other $p/N_{0}$ and $K$ values.
\end{table}

\section{Conclusions}
\label{Conclusions}
The performance of a multiuser communication system with single-antenna transmitters and a multi-antenna receiver was analyzed. The spatial mode-of-operation was investigated and more specifically the scenario when ZF or the more sophisticated ZF-SIC is applied at the receiver. Rank-$1$ Rician fading channels with i.n.i.d. statistics for each user were considered. This paper focused on a practical communications system with imperfections during the communication link, namely, imperfect channel estimation at the receiver and impaired transceiver hardware. Regarding the latter hardware imperfections, compensation algorithms are applied that mitigate the main hardware impairments. A new closed-form expression for the outage probability was derived, while a new analytical formula with the respect to the average ergodic capacity was obtained. Based on the results, there were some useful engineering observations: the impact of the mentioned imperfections to the system performance, the definition of asymptotic performance bounds, and the beneficial role of SIC-enabled reception in the presence of a LOS signal propagation.

\appendix
\subsection{Derivation of (\ref{snr})}
\label{appsnr}
\numberwithin{equation}{subsection}
\setcounter{equation}{0}
To proceed with the analytical derivation of SNR, the covariance matrix of $\mathbf{w}'$ is required, which is computed as
\begin{align}
\nonumber
&\mathbb{E}[\mathbf{w}'\mathbf{w}'^{\mathcal{H}}]\\
\nonumber
&=p\kappa^{2}_{T}\mathbf{I}_{M}+\sigma^{2}\mathbf{H}^{\dagger}\mathbf{\Omega}p(1+\kappa^{2}_{T})\mathbf{I}_{M}\mathbf{\Omega}^{\mathcal{H}}\left(\mathbf{H}^{\dagger}\right)^{\mathcal{H}}\\
\nonumber
&+\mathbf{H}^{\dagger}(p\kappa^{2}_{R}M+N_{0})\mathbf{I}_{N}\mathbf{H}^{\dagger}-\sigma \mathbf{H}^{\dagger} (p\kappa^{2}_{R}M+N_{0})\mathbf{I}_{N} \left(\mathbf{H}^{\dagger}\right)^{\mathcal{H}}\\
\nonumber
&\times \mathbf{\Omega}^{\mathcal{H}}\left(\mathbf{H}^{\dagger}\right)^{\mathcal{H}}+\sigma^{2}\mathbf{H}^{\dagger}\mathbf{\Omega}\mathbf{H}^{\dagger}(p\kappa^{2}_{R}M+N_{0})\mathbf{I}_{N}\left(\mathbf{H}^{\dagger}\right)^{\mathcal{H}}\mathbf{\Omega}^{\mathcal{H}}\\
\nonumber
&\times \left(\mathbf{H}^{\dagger}\right)^{\mathcal{H}}-\sigma \mathbf{H}^{\dagger}\mathbf{\Omega}\mathbf{H}^{\dagger}(p\kappa^{2}_{R}M+N_{0})\mathbf{I}_{N}\left(\mathbf{H}^{\dagger}\right)^{\mathcal{H}}\\
\nonumber
&\overset{\text{(a)}}=p\kappa^{2}_{T}\mathbf{I}_{M}+(p\kappa^{2}_{R}M+N_{0})(\mathbf{H}^{\mathcal{H}}\mathbf{H})^{-1}+\sigma^{2}p M (1+\kappa^{2}_{T})\\
\nonumber
&\times (\mathbf{H}^{\mathcal{H}}\mathbf{H})^{-1}+\sigma^{2}(p\kappa^{2}_{R}M+N_{0})\mathbf{H}^{\dagger}\mathbb{E}\left[\mathbf{\Omega}\mathbf{H}^{\dagger}\left(\mathbf{H}^{\dagger}\right)^{\mathcal{H}}\mathbf{\Omega}^{\mathcal{H}}\right]\left(\mathbf{H}^{\dagger}\right)^{\mathcal{H}}\\
\nonumber
&\overset{\text{(b)}}=p\Bigg\{\kappa^{2}_{T}\mathbf{I}_{M}+\bigg[\left(\kappa^{2}_{R}M+\frac{N_{0}}{p}\right)+\sigma^{2}M(1+\kappa^{2}_{T})\\
&+\left(\sigma^{2}\left(\kappa^{2}_{R}M+\frac{N_{0}}{p}\right)\Tr\left[(\mathbf{H}^{\mathcal{H}}\mathbf{H})^{-1}\right]\right)\bigg](\mathbf{H}^{\mathcal{H}}\mathbf{H})^{-1}\Bigg\},
\label{cov}
\end{align}
where in steps (a) and (b) we, respectively, used the equalities $\mathbf{H}^{\dagger}(\mathbf{H}^{\dagger})^{\mathcal{H}}=(\mathbf{H}^{\mathcal{H}}\mathbf{H})^{-1}$, $\mathbb{E}[\mathbf{\Omega}\mathbf{\Omega}^{\mathcal{H}}]=M\mathbf{I}_{N}$ and $\mathbb{E}[\mathbf{\Omega}\mathbf{H}^{\dagger}(\mathbf{H}^{\dagger})^{\mathcal{H}}\mathbf{\Omega}^{\mathcal{H}}]=\Tr[(\mathbf{H}^{\mathcal{H}}\mathbf{H})^{-1}]\mathbf{I}_{N}$. Using (\ref{cov}), (\ref{snr}) can be directly obtained.

\subsection{Derivation of (\ref{pdfyy}) and (\ref{pdfy})}
\label{appncxpdf}
\numberwithin{equation}{subsection}
\setcounter{equation}{0}
Dealing with ZF-SIC reception (including typical ZF as a special case), recall that $\mathbf{H} \in \mathbb{C}^{N\times (M-i+1)}$ at the $i$th SIC stage, as specified in the definition of Section \ref{Statistics of SNR}. This is due to the fact that during detection/decoding at the $i$th SIC stage, all the previous channel impact from $(i-1)$ stages has already been removed. We start by using the properties of matrix determinants in $\mathcal{Y}_{i}$ as
\begin{align}
\nonumber
\mathcal{Y}_{i}&=\frac{1}{[(\mathbf{H}^{\mathcal{H}}\mathbf{H})^{-1}]_{ii}}=\frac{\det[\mathbf{H}^{\mathcal{H}}\mathbf{H}]}{\det[\mathbf{H}_{i}^{\mathcal{H}}\mathbf{H}_{i}]}\\
&=\mathbf{h}^{\mathcal{H}}_{i}\mathbf{h}_{i}-\mathbf{h}^{\mathcal{H}}_{i}\mathbf{H}_{i}(\mathbf{H}^{\mathcal{H}}_{i}\mathbf{H}_{i})^{-1}\mathbf{H}^{\mathcal{H}}_{i}\mathbf{h}_{i}=\mathbf{h}^{\mathcal{H}}_{i}(\mathbf{I}_{N}-\mathbf{G}_{i})\mathbf{h}_{i},
\label{y1}
\end{align}
where $\mathbf{G}_{i}\triangleq \mathbf{H}_{i}(\mathbf{H}^{\mathcal{H}}_{i}\mathbf{H}_{i})^{-1}\mathbf{H}^{\mathcal{H}}_{i}$ and $\mathbf{H}_{i}\triangleq [\mathbf{h}_{1}\cdots\mathbf{h}_{i-1}\ \mathbf{h}_{i+1}\cdots\mathbf{h}_{M}]$ denoting the deflated version of $\mathbf{H}$ with its $i$th column (i.e., $\mathbf{h}_{i}$) removed.

Notice that $\mathbf{Q}_{i}\triangleq (\mathbf{I}_{N}-\mathbf{G}_{i})$ is a $N\times N$ matrix and represents the projection onto the null-space of $\mathbf{H}^{\mathcal{H}}_{i}$. In addition, it is a Hermitian, idempotent and symmetric matrix. Therefore, its eigenvalues are either zero or one. Particularly, they are used as
\begin{align}
\text{Eigenvalues of }\mathbf{Q}_{i}: \underbrace{0,0,\ldots,0}_{M-i},\underbrace{1,1,\ldots,1}_{N-M+i}.
\label{eigen}
\end{align}
Hence, $\mathbf{Q}_{i}$ has a rank equal to $N-M+i$, while it is statistically independent from $\mathbf{h}_{i}$.

Capitalizing on the latter observations and using the eigenvalue decomposition, the quadratic form of SNR in (\ref{y1}) can be further simplified as
\begin{align}
\mathcal{Y}_{i}=\mathbf{h}^{\mathcal{H}}_{i}\mathbf{Q}_{i}\mathbf{h}_{i}=\mathbf{h}^{\mathcal{H}}_{i}\mathbf{P}_{i}\mathbf{\Lambda}_{i}\mathbf{P}^{\mathcal{H}}_{i}\mathbf{h}_{i},
\label{y2}
\end{align}
where $\mathbf{P}_{i}$ denotes an orthogonal (unitary) matrix satisfying $\mathbf{P}_{i}\mathbf{P}^{\mathcal{H}}_{i}=\mathbf{P}^{\mathcal{H}}_{i}\mathbf{P}_{i}=\mathbf{I}_{N}$ and $\mathbf{\Lambda}_{i}=\diag\{\lambda_{1},\ldots,\lambda_{N}\}$ corresponds to the eigenvalues of $\mathbf{Q}_{i}$. Finally, based on (\ref{eigen}), (\ref{y2}) becomes
\begin{align}
\mathcal{Y}_{i}=\sum^{N}_{k=1}\lambda_{k}(\mathbf{P}^{\mathcal{H}}_{i}\mathbf{h}_{i})^{\mathcal{H}}(\mathbf{P}^{\mathcal{H}}_{i}\mathbf{h}_{i})=\sum^{N-M+i}_{k=1}(\mathbf{P}^{\mathcal{H}}_{i}\mathbf{h}_{i})^{\mathcal{H}}_{k}(\mathbf{P}^{\mathcal{H}}_{i,k}\mathbf{h}_{i})_{k},
\label{y3}
\end{align}
where $(\mathbf{P}^{\mathcal{H}}_{i}\mathbf{h}_{i})_{k}$ stands for the $k$th coefficient of vector $(\mathbf{P}^{\mathcal{H}}_{i}\mathbf{h}_{i})$. In the trivial case when $\mathbf{h}_{i}$ is zero-mean, then $\mathbf{P}^{\mathcal{H}}_{i}\mathbf{h}_{i}\overset{\text{d}}=\mathbf{h}_{i}$ (i.e., isotropically distributed). This yields $\mathcal{Y}_{i}\overset{\text{d}}=(d^{-\alpha_{i}}_{i}/(K+1))\chi^{2}_{N-M+i}$. On the other hand, when $\mathbf{h}_{i}$ is a non zero-mean vector, i.e., Rician-distributed, the former isotropic identity does not hold. Fortunately, it was recently indicated in \cite[Eq. (8) and \S 4]{siriteanu2016chi} that (\ref{y3}) follows a conditional non-central chi-squared distribution, in the case of a non zero-mean $\mathbf{h}_{i}$, while its corresponding non-centrality parameter follows a central Beta distribution, which is independent of $\mathcal{Y}_{i}$. This yields (\ref{pdfyy}) and (\ref{pdfy}).

\subsection{Derivation of (\ref{cdfyyyyyy}) and (\ref{cdfyyyy})}
\label{appcdfclosform}
\numberwithin{equation}{subsection}
\setcounter{equation}{0}
For $i=1$ (i.e., the Rician-faded stream), we get \cite{b:marcum}
\begin{align}
\nonumber
F_{\mathcal{Y}_{1}|\theta_{1}}(x|u)&=1-\int^{\infty}_{x}f_{\mathcal{Y}_{1}|\theta_{1}}(y|u)dy\\
&=1-Q_{N-M+1}\left(\sqrt{u},\sqrt{\frac{(K+1)x}{d^{-\alpha_{1}}_{1}}}\right).
\label{just}
\end{align}
Thus, the corresponding unconditional CDF reads as
\begin{align}
\nonumber
F_{\mathcal{Y}_{1}}(y)&=1-\int^{1}_{0}Q_{N-M+1}\left(\sqrt{u},\sqrt{\frac{(K+1)x}{d^{-\alpha_{1}}_{1}}}\right)\\
\nonumber
&\times \frac{u^{N-M}(1-u)^{M-2}}{B(N-M+1,M-1)}du\\
\nonumber
&=1-\frac{\sum^{M-2}_{j=0}\frac{\binom{M-2}{j}}{(j+1)}}{B(N-M+1,M-1)}\\
&\times \int^{1}_{0}u^{N-M+j}Q_{N-M+1}\left(\sqrt{u},\sqrt{\frac{(K+1)x}{d^{-\alpha_{1}}_{1}}}\right)du.
\label{just1}
\end{align}
To capture (\ref{just1}) in a closed form, an integral of the type $\mathcal{J}\triangleq \int^{1}_{0}u^{a}Q_{m}(\sqrt{u},b)du$ needs to be solved. Implementing integration by parts, it follows that
\begin{align}
\nonumber
&\mathcal{J}=\int^{1}_{0}u^{a}Q_{m}(\sqrt{u},b)du\\
\nonumber
&=\left.Q_{m}(\sqrt{u},b)\frac{u^{a+1}}{(a+1)}\right\vert ^{1}_{0}-\int^{1}_{0}\frac{u^{a+1}}{(a+1)}\frac{\partial Q_{m}(\sqrt{u},b)}{\partial u}du\\
\nonumber
&\overset{\text{(a)}}=\frac{Q_{m}(1,b)}{(a+1)}-\int^{1}_{0}\frac{u^{a+1}}{2(a+1)}\left(Q_{m+1}(\sqrt{u},b)-Q_{m}(\sqrt{u},b)\right)du\\
\nonumber
&\overset{\text{(b)}}=\frac{Q_{m}(1,b)}{(a+1)}-\int^{1}_{0}\frac{u^{a+\frac{1}{2}}}{2(a+1)}b^{m}\exp\left(-\frac{(b^{2}+u)}{2}\right)I_{m}(b\sqrt{u})du\\
\nonumber
&=\frac{Q_{m}(1,b)}{(a+1)}-\bigg[\int^{\infty}_{0}\frac{u^{2 a+2}}{(a+1)}b^{m}\exp\left(-\frac{(b^{2}+u^{2})}{2}\right)I_{m}(bu)du\\
&-\int^{\infty}_{1}\frac{u^{2 a+2}}{(a+1)}b^{m}\exp\left(-\frac{(b^{2}+u^{2})}{2}\right)I_{m}(bu)du\bigg],
\label{just2}
\end{align}
where (a) and (b) arise due to \cite[Eqs. (2) and (16)]{j:brychkov2012some}. Then, by definition, (\ref{just2}) reads as\footnote{To our knowledge, the derived expression in (\ref{just2}) is novel and has not been reported elsewhere into the literature. Note that the integral in $\mathcal{J}$ cannot be considered as a special case of \cite[App. D]{j:MartinezGoldsmith}.}
\begin{align}
\nonumber
\mathcal{J}&=\frac{1}{(a+1)}\\
&\times \left\{Q_{m}(1,b)-b^{m}[Q_{2(a+1),m}(b,0)-Q_{2(a+1),m}(b,1)]\right\}.
\label{just3}
\end{align}
Using (\ref{just3}) into (\ref{just1}), the desired result in (\ref{cdfyyyyyy}) is obtained.

For $2\leq i\leq M$, noticing that $F_{\mathcal{Y}_{i}}(y)=1-\int^{\infty}_{x}f_{\mathcal{Y}_{i}}(y)dy$, it yields (\ref{cdfyyyy}).

\subsection{Derivation of (\ref{cdfy})}
\label{appcdfy}
\numberwithin{equation}{subsection}
\setcounter{equation}{0}
The CDF of SNR for the $i$th stream is defined as
\begin{align}
F_{\text{SNR}_{i}}(x)\triangleq \text{Pr}\left[\text{SNR}_{i}\leq x\right]=\text{Pr}\left[\mathcal{Y}_{i}\leq \frac{\psi x}{1-\kappa^{2}_{T}x}\right].
\label{cdfdef}
\end{align}
Referring back to (\ref{psi}), the cumbersome parameter $\Tr[(\mathbf{H}^{\mathcal{H}}\mathbf{H})^{-1}]$ is included within the auxiliary variable $\psi$. The analytical representation of $\Tr[(\mathbf{H}^{\mathcal{H}}\mathbf{H})^{-1}]$ is infeasible for the considered rank-$1$ Rician fading channel. This occurs because the inverse non-central Wishart PDF is involved, which is generally unknown so far. Nevertheless, it can be efficiently approximated by a corresponding inverse central Wishart PDF. More specifically, let
\begin{align}
\mathbf{Z}\overset{\text{d}}=\mathcal{CW}_{M}\left(N-M+i,\diag\left\{\frac{d^{-a_{i}}_{i}}{(K+1)}\right\}^{M}_{i=1}+\frac{1}{N}\mathbf{H}^{\mathcal{H}}_{d}\mathbf{H}_{d}\right).
\end{align}
Then, the Gramian matrix $\mathbf{H}^{\mathcal{H}}\mathbf{H}$ can be approached by a central Wishart distribution, such that
\begin{align}
\mathbf{H}^{\mathcal{H}}\mathbf{H}\overset{\text{d}}\approx \mathbf{Z}.
\end{align}
It was indicated in \cite{tan1983approximating} that the Gramian $\mathbf{H}^{\mathcal{H}}\mathbf{H}$ and $\mathbf{Z}$ share the same expected value, while there is a slight difference between their variances in the order of $\mathcal{O}(N^{-1})$. Moreover, the accuracy of the aforementioned approximation for a rank-$1$ Gramian $\mathbf{H}^{\mathcal{H}}\mathbf{H}$ was tested and verified in \cite[\S VI.B.2]{j:SiriteanuMiyanaga2012}. 

Therefore, $\Tr[(\mathbf{H}^{\mathcal{H}}\mathbf{H})^{-1}]\approx \Tr[\mathbf{Z}^{-1}]$, while it holds from \cite[Lemma 6]{LozanoTulinoVerdu} that
\begin{align}
\mathbb{E}\left[\Tr[(\mathbf{H}^{\mathcal{H}}\mathbf{H})^{-1}]\right]\approx \mathbb{E}\left[\Tr[\mathbf{Z}^{-1}]\right]=\frac{M-i+1}{N-M+i-1},\ \ N>M
\label{mean}
\end{align}
and
\begin{align}
\nonumber
&\Var\left[\Tr[(\mathbf{H}^{\mathcal{H}}\mathbf{H})^{-1}]\right]\approx \Var\left[\Tr[\mathbf{Z}^{-1}]\right]\\
&=\frac{(M-i+1)N}{(N-M+i-1)^{2}((N-M+i-1)^{2}-1)},N>M+1.
\label{var}
\end{align}
Obviously, both the above expectation and the corresponding variance (i.e., second order statistic) take very low values for the considered case study, especially when $N\gg M$. Thus, $\Tr[(\mathbf{H}^{\mathcal{H}}\mathbf{H})^{-1}]\ll 1$. Furthermore, recall that $\{\kappa^{2}_{T},\kappa^{2}_{R}\}\ll 1$. Consequently, keeping in mind that base stations are normally equipped with advanced low-noise amplifiers (LNAs), while it usually holds that $\sigma^{2} \ll 1$ \cite[Fig. 3]{Bjornsonhi2014}, the parameter $\sigma^{2}(\kappa^{2}_{R}M+N_{0}/p)\Tr[(\mathbf{H}^{\mathcal{H}}\mathbf{H})^{-1}]$ can be neglected from $\psi$, thus, the latter term can be relaxed as
\begin{align}
\psi\approx \left(\kappa^{2}_{R}M+\frac{N_{0}}{p}\right)+\sigma^{2}M(1+\kappa^{2}_{T}).
\label{psiapprox}
\end{align}
Using (\ref{psiapprox}) into (\ref{cdfdef}) yields (\ref{cdfy}).

\subsection{Derivation of (\ref{ergcap})}
\label{appergcap}
\numberwithin{equation}{subsection}
\setcounter{equation}{0}
Using (\ref{out}) into (\ref{cdef}), while setting $\kappa_{T}=0$, it yields that
\begin{align}
\nonumber
&\overline{C}_{i|\theta_{i}}=\\
&\int^{\infty}_{0}\frac{Q_{N-M+i}\left(\sqrt{\theta_{i}},\sqrt{\frac{\left(\left(\kappa^{2}_{R}M+\frac{N_{0}}{p}+\sigma^{2}M(1+\kappa^{2}_{T})\right) x\right)}{(K+1)^{-1}d^{-\alpha_{i}}_{i}}}\right)}{1+x}dx.
\label{cdefff}
\end{align}
To further proceed, the Marcum-$Q$ function can be expanded as \cite{j:Dillard1973}
\begin{align}
\nonumber
&Q_{m}(\sqrt{a},\sqrt{b x})=\\
&\sum^{\infty}_{k=0}\sum^{m+k-1}_{j=0}\frac{\left(\frac{a}{2}\right)^{k}\left(\frac{b x}{2}\right)^{j}\exp\left(-\frac{b x}{2}\right)}{k!j!\exp\left(\frac{a}{2}\right)}, \ \ m \in \mathbb{N}^{+},
\label{qexpansion}
\end{align}
and
\begin{align}
Q_{m}(0,\sqrt{b x})=\sum^{m-1}_{j=0}\frac{\left(\frac{bx}{2}\right)^{j}\exp\left(-\frac{b x}{2}\right)}{j!}, \ \ m \in \mathbb{N}^{+}.
\label{qexpansion1}
\end{align}
Then, an integral of the following type appears
\begin{align}
\int^{\infty}_{0}\frac{x^j \exp\left(-\frac{b x}{2}\right)}{1+x}dx,
\label{intqfunc}
\end{align}
which can be directly evaluated in a closed-form solution with the aid of \cite[Eq. (3.383.10)]{tables}. Thereby, the desired result in (\ref{ergcapray}) is extracted for the Rayleigh-fading case, using (\ref{cdefff}), (\ref{qexpansion1}) and (\ref{intqfunc}).

For the more challenging Rician-fading case, the conditioning on the Beta-distributed $\theta_{1}$ parameter needs also to be averaged out. Hence, following similar steps as in deriving (\ref{ergcapray}) and using (\ref{qexpansion}), we get
\begin{align}
\nonumber
\overline{C}_{1}&=\int^{1}_{0}\overline{C}_{1|\theta_{1}}(u)f_{\theta_{1}}(u)du\\
\nonumber
&=\sum^{\infty}_{k=0}\sum^{N-M+k}_{j=0}\left(\frac{\left(\kappa^{2}_{R}M+\frac{N_{0}}{p}+\sigma^{2}M\right)}{2(K+1)^{-1}d^{-\alpha_{1}}_{1}}\right)^{j}\\
\nonumber
&\times \frac{\exp\left(\frac{\left(\kappa^{2}_{R}M+\frac{N_{0}}{p}+\sigma^{2}M\right)}{2(K+1)^{-1}d^{-\alpha_{1}}_{1}}\right)\Gamma\left(-j,\frac{\left(\kappa^{2}_{R}M+\frac{N_{0}}{p}+\sigma^{2}M\right)}{2(K+1)^{-1}d^{-\alpha_{1}}_{1}}\right)}{k!2^{k}B(N-M+1,M-1)}\\
&\times \int^{1}_{0}u^{N-M+k}(1-u)^{M-2}\exp\left(-\frac{u}{2}\right)du.
\end{align}
After some straightforward algebra, we arrive at (\ref{ergcap}).


\ifCLASSOPTIONcaptionsoff
  \newpage
\fi

\bibliographystyle{IEEEtran}
\bibliography{IEEEabrv,References}

\begin{thebibliography}{10}
\providecommand{\url}[1]{#1}
\csname url@samestyle\endcsname
\providecommand{\newblock}{\relax}
\providecommand{\bibinfo}[2]{#2}
\providecommand{\BIBentrySTDinterwordspacing}{\spaceskip=0pt\relax}
\providecommand{\BIBentryALTinterwordstretchfactor}{4}
\providecommand{\BIBentryALTinterwordspacing}{\spaceskip=\fontdimen2\font plus
\BIBentryALTinterwordstretchfactor\fontdimen3\font minus
  \fontdimen4\font\relax}
\providecommand{\BIBforeignlanguage}[2]{{%
\expandafter\ifx\csname l@#1\endcsname\relax
\typeout{** WARNING: IEEEtran.bst: No hyphenation pattern has been}%
\typeout{** loaded for the language `#1'. Using the pattern for}%
\typeout{** the default language instead.}%
\else
\language=\csname l@#1\endcsname
\fi
#2}}
\providecommand{\BIBdecl}{\relax}
\BIBdecl

\bibitem{j:gesbert}
D.~Gesbert, ``Robust linear {MIMO} receivers: a minimum error-rate approach,''
  \emph{{IEEE} Trans. Signal Process.}, vol.~51, no.~11, pp. 2863--2871, Nov.
  2003.

\bibitem{j:MiridakisSurvey}
N.~Miridakis and D.~Vergados, ``A survey on the successive interference
  cancellation performance for single-antenna and multiple-antenna {OFDM}
  systems,'' \emph{{IEEE} Commun. Surveys Tuts.}, vol.~15, no.~1, pp. 312--335,
  First 2013.

\bibitem{golden1999detection}
G.~D. Golden, C.~J. Foschini, R.~Valenzuela, P.~W. Wolniansky \emph{et~al.},
  ``Detection algorithm and initial laboratory results using {V-BLAST}
  space-time communication architecture,'' \emph{Electron. Lett.}, vol.~35,
  no.~1, pp. 14--16, 1999.

\bibitem{j:Toboso2014}
A.~U. Toboso, S.~Loyka, and F.~Gagnon, ``Optimal detection ordering for coded
  {V-BLAST},'' \emph{{IEEE} Trans. Commun.}, vol.~62, no.~1, pp. 100--111, Jan.
  2014.

\bibitem{j:MiridakisKaragiannidis2014}
N.~I. Miridakis, M.~Matthaiou, and G.~K. Karagiannidis, ``Multiuser relaying
  over mixed {RF/FSO} links,'' \emph{{IEEE} Trans. Commun.}, vol.~62, no.~5,
  pp. 1634--1645, May 2014.

\bibitem{j:JiangVaranasi2008}
Y.~Jiang and M.~K. Varanasi, ``Spatial multiplexing architectures with jointly
  designed rate-tailoring and ordered {BLAST} decoding - part i:
  Diversity-multiplexing tradeoff analysis,'' \emph{{IEEE} Trans. Wireless
  Commun.}, vol.~7, no.~8, pp. 3252--3261, Aug. 2008.

\bibitem{j:JiangVaranasi2011}
Y.~Jiang, M.~K. Varanasi, and J.~Li, ``Performance analysis of {ZF} and {MMSE}
  equalizers for {MIMO} systems: An in-depth study of the high {SNR} regime,''
  \emph{{IEEE} Trans. Inf. Theory}, vol.~57, no.~4, pp. 2008--2026, Apr. 2011.

\bibitem{j:Loyka2006}
S.~Loyka and F.~Gagnon, ``{V-BLAST} without optimal ordering: analytical
  performance evaluation for {R}ayleigh fading channels,'' \emph{{IEEE} Trans.
  Commun.}, vol.~54, no.~6, pp. 1109--1120, Jun. 2006.

\bibitem{j:LoykaGagnon2004}
------, ``Performance analysis of the {V-BLAST} algorithm: an analytical
  approach,'' \emph{{IEEE} Trans. Wireless Commun.}, vol.~3, no.~4, pp.
  1326--1337, Jul. 2004.

\bibitem{j:WCLMiridakisVergados2013}
N.~I. Miridakis and D.~D. Vergados, ``Performance analysis of the ordered
  {V-BLAST} approach over {N}akagami-{$m$} fading channels,'' \emph{{IEEE}
  {W}ireless {C}ommun. {L}ett.}, vol.~2, no.~1, pp. 18--21, Feb. 2013.

\bibitem{j:OzyurtTorlak2013}
S.~Ozyurt and M.~Torlak, ``Exact joint distribution analysis of zero-forcing
  {V-BLAST} gains with greedy ordering,'' \emph{{IEEE} Trans. Wireless
  Commun.}, vol.~12, no.~11, pp. 5377--5385, Nov. 2013.

\bibitem{schenk2008rf}
T.~Schenk, \emph{RF Imperfections in High-Rate Wireless Systems: Impact and
  Digital Compensation}.\hskip 1em plus 0.5em minus 0.4em\relax Springer
  Science \& Business Media, 2008.

\bibitem{Bjornsonhi2014}
E.~Bjornson, J.~Hoydis, M.~Kountouris, and M.~Debbah, ``Massive {MIMO} systems
  with non-ideal hardware: Energy efficiency, estimation, and capacity
  limits,'' \emph{{IEEE} Trans. Inf. Theory}, vol.~60, no.~11, pp. 7112--7139,
  Nov. 2014.

\bibitem{zetterberg2011experimental}
P.~Zetterberg, ``Experimental investigation of {TDD} reciprocity-based
  zero-forcing transmit precoding,'' \emph{EURASIP J. Adv. Signal Process.},
  vol. 2011, p.~5, 2011.

\bibitem{c:StuderWenk2010}
C.~Studer, M.~Wenk, and A.~Burg, ``{MIMO} transmission with residual
  transmit-{RF} impairments,'' in \emph{Proc. ITG/IEEE Work. Smart Ant. (WSA)},
  Feb. 2010, pp. 189--196.

\bibitem{c:ZhangMatthaiouBjornson2014}
X.~Zhang, M.~Matthaiou, M.~Coldrey, and E.~Bjornson, ``Energy efficiency
  optimization in hardware-constrained large-scale {MIMO} systems,'' in
  \emph{Proc. 11th Int. Symp. Wireless Commun. Syst. (ISWCS)}, Aug. 2014, pp.
  992--996.

\bibitem{j:WangMurch2007}
C.~Wang, E.~K.~S. Au, R.~D. Murch, W.~H. Mow, R.~S. Cheng, and V.~Lau, ``On the
  performance of the {MIMO} zero-forcing receiver in the presence of channel
  estimation error,'' \emph{{IEEE} Trans. Wireless Commun.}, vol.~6, no.~3, pp.
  805--810, Mar. 2007.

\bibitem{j:Narasimhan2005}
R.~Narasimhan, ``Error propagation analysis of {V-BLAST} with
  channel-estimation errors,'' \emph{{IEEE} Trans. Commun.}, vol.~53, no.~1,
  pp. 27--31, Jan. 2005.

\bibitem{kyosti2007winner}
P.~Ky{\"o}sti, J.~Meinil{\"a}, L.~Hentil{\"a}, X.~Zhao, T.~J{\"a}ms{\"a},
  C.~Schneider, M.~Narandzi, M.~Milojevi, A.~Hong, J.~Ylitalo \emph{et~al.},
  ``{WINNER II} {C}hannel {M}odels. part {I},'' IST-Information Society
  Technologies, Tech. Rep. IST-4-027756 WINNER II D, Tech. Rep., 2007.

\bibitem{j:RusekPersson2013}
F.~Rusek, D.~Persson, B.~K. Lau, E.~G. Larsson, T.~L. Marzetta, O.~Edfors, and
  F.~Tufvesson, ``Scaling up {MIMO}: Opportunities and challenges with very
  large arrays,'' \emph{{IEEE} Signal Process. Mag.}, vol.~30, no.~1, pp.
  40--60, Jan. 2013.

\bibitem{j:RappaportShu2013}
T.~S. Rappaport, S.~Sun, R.~Mayzus, H.~Zhao, Y.~Azar, K.~Wang, G.~N. Wong,
  J.~K. Schulz, M.~Samimi, and F.~Gutierrez, ``Millimeter wave mobile
  communications for 5{G} cellular: It will work!'' \emph{{IEEE} {A}ccess},
  vol.~1, pp. 335--349, 2013.

\bibitem{j:AndrewsBuzzi2014}
J.~Andrews, S.~Buzzi, W.~Choi, S.~V. Hanly, A.~Lozano, A.~C.~K. Soong, and
  J.~C. Zhang, ``What will 5{G} be?'' \emph{{IEEE} J. Sel. Areas Commun.},
  vol.~32, no.~6, pp. 1065--1082, Jun. 2014.

\bibitem{j:JinMatthaiou2015}
S.~Jin, W.~Tan, M.~Matthaiou, J.~Wang, and K.~K. Wong, ``Statistical eigenmode
  transmission for the {MU-MIMO} downlink in {R}ician fading,'' \emph{{IEEE}
  Trans. Wireless Commun.}, vol.~14, no.~12, pp. 6650--6663, Dec. 2015.

\bibitem{j:PengWang2015}
Y.~Peng, Y.~Li, and P.~Wang, ``An enhanced channel estimation method for
  millimeter wave systems with massive antenna arrays,'' \emph{{IEEE} Commun.
  Lett.}, vol.~19, no.~9, pp. 1592--1595, Sep. 2015.

\bibitem{j:Siriteanurician2014}
C.~Siriteanu, A.~Takemura, S.~Kuriki, D.~S.~P. Richards, and H.~Shin, ``Schur
  complement based analysis of {MIMO} zero-forcing for {R}ician fading,''
  \emph{{IEEE} Trans. Wireless Commun.}, vol.~14, no.~4, pp. 1757--1771, Apr.
  2015.

\bibitem{j:siriteanuexact2014}
C.~Siriteanu, S.~D. Blostein, A.~Takemura, H.~Shin, S.~Yousefi, and S.~Kuriki,
  ``Exact {MIMO} zero-forcing detection analysis for transmit-correlated
  {R}ician fading,'' \emph{{IEEE} Trans. Wireless Commun.}, vol.~13, no.~3, pp.
  1514--1527, Mar. 2014.

\bibitem{j:SiriteanuMiyanaga2012}
C.~Siriteanu, Y.~Miyanaga, S.~D. Blostein, S.~Kuriki, and X.~Shi, ``{MIMO}
  zero-forcing detection analysis for correlated and estimated {R}ician
  fading,'' \emph{{IEEE} Trans. Veh. Technol.}, vol.~61, no.~7, pp. 3087--3099,
  Sep. 2012.

\bibitem{j:Zhang2014}
Q.~Zhang, S.~Jin, K.-K. Wong, H.~Zhu, and M.~Matthaiou, ``Power scaling of
  uplink massive {MIMO} systems with arbitrary-rank channel means,''
  \emph{{IEEE} J. Sel. Topics Signal Process.}, vol.~8, no.~5, pp. 966--981,
  Oct. 2014.

\bibitem{j:XueSellathurai2015}
J.~Xue, M.~Sellathurai, T.~Ratnarajah, and Z.~Ding, ``Performance analysis for
  multi-way relaying in {R}ician fading channels,'' \emph{{IEEE} Trans.
  Commun.}, vol.~63, no.~11, pp. 4050--4062, Nov. 2015.

\bibitem{j:Siriteanurank0ne}
C.~Siriteanu, A.~Takemura, C.~Koutschan, S.~Kuriki, D.~S.~P. Richards, and
  H.~Shin, ``Exact {ZF} analysis and computer-algebra-aided evaluation in
  rank-$1$ {L}o{S} {R}ician fading,'' \emph{{IEEE} Trans. Wireless Commun.}, To
  be published, 2016.

\bibitem{j:SaquibHossain}
N.~Saquib, E.~Hossain, L.~B. Le, and D.~I. Kim, ``Interference management in
  {OFDMA} femtocell networks: issues and approaches,'' \emph{{IEEE} Wireless
  Commun. Mag.}, vol.~19, no.~3, pp. 86--95, Jun. 2012.

\bibitem{j:BradyBehdad}
J.~Brady, N.~Behdad, and A.~M. Sayeed, ``Beamspace {MIMO} for millimeter-wave
  communications: System architecture, modeling, analysis, and measurements,''
  \emph{{IEEE} Trans. Antennas Propag.}, vol.~61, no.~7, pp. 3814--3827, Jul.
  2013.

\bibitem{tables}
I.~S. Gradshteyn and I.~M. Ryzhik, \emph{Table of Integrals, Series, and
  Products}.\hskip 1em plus 0.5em minus 0.4em\relax Academic Press, 2007.

\bibitem{b:marcum}
J.~I. Marcum, \emph{Table of {Q}-functions}.\hskip 1em plus 0.5em minus
  0.4em\relax U.S. Air Force Project RAND Res. Memo. M-339, ASTIA document AD
  1165451, Santa Monica, CA, 1950.

\bibitem{nuttall1972some}
A.~H. Nuttall, ``Some integrals involving the {Q}-function,'' DTIC Document,
  Tech. Rep., 1972.

\bibitem{j:Heath2016}
R.~W. Heath, N.~González-Prelcic, S.~Rangan, W.~Roh, and A.~M. Sayeed, ``An
  overview of signal processing techniques for millimeter wave {MIMO}
  systems,'' \emph{{IEEE} J. Sel. Topics Signal Process.}, vol.~10, no.~3, pp.
  436--453, Apr. 2016.

\bibitem{j:Barati2015}
C.~N. Barati, S.~A. Hosseini, S.~Rangan, P.~Liu, T.~Korakis, S.~S. Panwar, and
  T.~S. Rappaport, ``Directional cell discovery in millimeter wave cellular
  networks,'' \emph{{IEEE} Trans. Wireless Commun.}, vol.~14, no.~12, pp.
  6664--6678, Dec. 2015.

\bibitem{j:Alkhateeb2014}
A.~Alkhateeb, O.~E. Ayach, G.~Leus, and R.~W. Heath, ``Channel estimation and
  hybrid precoding for millimeter wave cellular systems,'' \emph{{IEEE} J. Sel.
  Topics Signal Process.}, vol.~8, no.~5, pp. 831--846, Oct. 2014.

\bibitem{j:Ayach2014}
O.~E. Ayach, S.~Rajagopal, S.~Abu-Surra, Z.~Pi, and R.~W. Heath, ``Spatially
  sparse precoding in millimeter wave {MIMO} systems,'' \emph{{IEEE} Trans.
  Wireless Commun.}, vol.~13, no.~3, pp. 1499--1513, Mar. 2014.

\bibitem{c:matthaiouhimimo2014}
X.~Zhang, M.~Matthaiou, E.~Bjornson, M.~Coldrey, and M.~Debbah, ``On the {MIMO}
  capacity with residual transceiver hardware impairments,'' in \emph{{IEEE}
  Int. Conf. Commun. (ICC)}, Sydney, NSW, Jun. 2014, pp. 5299--5305.

\bibitem{b:wenk2010mimo}
M.~Wenk, \emph{{MIMO-OFDM} Testbed: Challenges, Implementations, and
  Measurement Results}.\hskip 1em plus 0.5em minus 0.4em\relax ser. Series in
  Microelectronics. ETH, 2010.

\bibitem{j:zhang2015achievable}
J.~Zhang, L.~Dai, X.~Zhang, E.~Bjornson, and Z.~Wang, ``Achievable rate of
  {R}ician large-scale {MIMO} channels with transceiver hardware impairments,''
  \emph{{IEEE} Trans. Veh. Technol.}, 2015, To be published.

\bibitem{b:holma2011lte}
H.~Holma and A.~Toskala, \emph{LTE for UMTS: Evolution to LTE-advanced}.\hskip
  1em plus 0.5em minus 0.4em\relax John Wiley \& Sons, 2011.

\bibitem{j:liu2006evm}
R.~Liu, Y.~Li, H.~Chen, and Z.~Wang, ``{EVM} estimation by analyzing
  transmitter imperfections mathematically and graphically,'' \emph{Analog
  Integrated Circuits and Signal Processing}, vol.~48, no.~3, pp. 257--262,
  2006.

\bibitem{j:Georgiadis2004evm}
A.~Georgiadis, ``Gain, phase imbalance, and phase noise effects on error vector
  magnitude,'' \emph{{IEEE} Trans. Veh. Technol.}, vol.~53, no.~2, pp.
  443--449, Mar. 2004.

\bibitem{b:goldsmith}
A.~J. Goldsmith, \emph{Wireless Communications}.\hskip 1em plus 0.5em minus
  0.4em\relax New York: Cambridge University Press, 2005.

\bibitem{b:johnson1994continuous}
N.~L. Johnson, S.~Kotz, and N.~Balakrishnan, \emph{Continuous Univariate
  Distributions, Vol. 2}.\hskip 1em plus 0.5em minus 0.4em\relax New York: John
  Wiley \& Sons, 1994.

\bibitem{j:GOREHeathPaulraj}
D.~A. Gore, R.~W. Heath, and A.~J. Paulraj, ``Transmit selection in spatial
  multiplexing systems,'' \emph{{IEEE} Commun. Lett.}, vol.~6, no.~11, pp.
  491--493, Nov. 2002.

\bibitem{c:Sofotasiosnuttall}
P.~C. Sofotasios and S.~Freear, ``A novel representation for the {N}uttall
  {Q}-function,'' in \emph{{IEEE} Int. Conf. Wireless Inf. Technol. Systems
  (ICWITS)}, Honolulu, Hawaii, Aug. 2010, pp. 1--4.

\bibitem{b:Balanis}
C.~A. Balanis, \emph{Antenna Theory: Analysis and Design, 3rd Edition}.\hskip
  1em plus 0.5em minus 0.4em\relax Wiley-Interscience, 2016.

\bibitem{siriteanu2016chi}
C.~Siriteanu, S.~Kuriki, D.~Richards, and A.~Takemura, ``Chi-square mixture
  representations for the distribution of the scalar {S}chur complement in a
  noncentral {W}ishart matrix,'' \emph{Statistics \& Probability Letters}, vol.
  115, pp. 79--87, 2016.

\bibitem{j:LuSwindlehurst2014}
L.~Lu, G.~Y. Li, A.~L. Swindlehurst, A.~Ashikhmin, and R.~Zhang, ``An overview
  of massive {MIMO}: Benefits and challenges,'' \emph{{IEEE} J. Sel. Topics
  Signal Process.}, vol.~8, no.~5, pp. 742--758, Oct. 2014.

\bibitem{j:GilAmparo2014}
A.~Gil, J.~Segura, and N.~M. Temme, ``The asymptotic and numerical inversion of
  the {M}arcum {$Q$}-function,'' \emph{Studies in Applied Mathematics}, vol.
  133, no.~2, pp. 257--278, 2014.

\bibitem{j:Kapinasnuttall}
V.~M. Kapinas, S.~K. Mihos, and G.~K. Karagiannidis, ``On the monotonicity of
  the generalized {M}arcum and {N}uttall {Q} -functions,'' \emph{{IEEE} Trans.
  Inf. Theory}, vol.~55, no.~8, pp. 3701--3710, Aug. 2009.

\bibitem{j:brychkov2012some}
Y.~A. Brychkov, ``On some properties of the {M}arcum {Q} function,''
  \emph{Integral Transforms and Special Functions}, vol.~23, no.~3, pp.
  177--182, 2012.

\bibitem{j:MartinezGoldsmith}
F.~J. Lopez-Martinez, E.~Martos-Naya, J.~F. Paris, and A.~Goldsmith,
  ``Eigenvalue dynamics of a central {W}ishart matrix with application to
  {MIMO} systems,'' \emph{{IEEE} Trans. Inf. Theory}, vol.~61, no.~5, pp.
  2693--2707, May 2015.

\bibitem{tan1983approximating}
W.~Y. Tan and R.~P. Gupta, ``On approximating a linear combination of central
  wishart matrices with positive coefficients,'' \emph{Commun. Statist.-Theory
  Meth.}, vol.~12, no.~22, pp. 2589--2600, 1983.

\bibitem{LozanoTulinoVerdu}
A.~Lozano, A.~M. Tulino, and S.~Verdu, ``Multiple-antenna capacity in the
  low-power regime,'' \emph{{IEEE} Trans. Inf. Theory}, vol.~49, no.~10, pp.
  2527--2544, Oct. 2003.

\bibitem{j:Dillard1973}
G.~M. Dillard, ``Recursive computation of the generalized {Q}-function,''
  \emph{{IEEE} Trans. Aerosp. Electron. Syst.}, vol. AES-9, no.~4, pp.
  614--615, Jul. 1973.

\end{thebibliography}

\end{document}